\documentclass[superscriptaddress, prx,aps, reprint,longbibliography]{revtex4-2}
\usepackage{amsmath}
\usepackage{enumitem}
\usepackage{amsthm}
\usepackage{amssymb}
\usepackage{mathtools}
\usepackage[linesnumbered,ruled]{algorithm2e}
\SetEndCharOfAlgoLine{}
\usepackage{svg}
\usepackage{subfloat}
\usepackage{placeins}
\usepackage[caption=false]{subfig}
\usepackage{wasysym}
\usepackage[hidelinks, colorlinks=true]{hyperref}
\usepackage{bm}
\usepackage{booktabs}
\usepackage{braket}
\usepackage{pifont}     
\usepackage{xcolor}
\usepackage{fontawesome}
\usepackage{multirow}   
\usepackage{hhline}     
\usepackage{array}      
\usepackage{graphicx} 

\newtheorem{theorem}{Theorem}
\newtheorem{proposition}{Proposition}
\newtheorem{definition}{Definition}
\newtheorem{corollary}{Corollary}

\newcommand{\thickhline}{\noalign{\hrule height 0.8pt}}

\newcolumntype{?}{!{\vrule width 0.8pt}}

\newcommand{\eq}[1]{Eq.~(\ref{eq:#1})}
\newcommand{\fig}[1]{Figure~\ref{fig:#1}}
\newcommand{\sect}[1]{Section~\ref{sec:#1}}
\newcommand{\tab}[1]{Table~\ref{table:#1}}
\newcommand{\app}[1]{Appendix~\ref{app:#1}}
\newcommand{\thm}[1]{Theorem~\ref{thm:#1}}
\newcommand{\propos}[1]{Proposition~\ref{propos:#1}}
\newcommand{\defn}[1]{Definition~\ref{def:#1}}
\newcommand{\corol}[1]{Corollary~\ref{corol:#1}}

\definecolor{myyellow}{rgb}{0.788, 0.62, 0.063}

\begin{document}

\title{High-performance syndrome extraction circuits for quantum codes}
\author{Armands Strikis}
\email{strikis@maths.ox.ac.uk}
\affiliation{Mathematical Institute, University of Oxford, Woodstock Road, Oxford OX2 6GG, United Kingdom}
\author{Dan E. Browne}
\affiliation{Department of Physics \& Astronomy, University College London, London, WC1E 6BT, United Kingdom}
\author{Michael E. Beverland}
\email{michael.beverland@gmail.com}
\affiliation{IBM Quantum}

\date{\today}

\begin{abstract}
We present a fast and effective framework for analysing and designing syndrome extraction circuits (SECs). 
Our approach is based on left--right circuits,
a general SEC design that maintains low depth and avoids gate interleaving constraints by staggering $X$ and $Z$ stabiliser checks.
Initially proposed for specific classes of codes~\cite{xu2024constant}, we generalise this construction to arbitrary CSS codes and optimise the circuit structure to achieve low qubit idling time, large effective distance, and reduced minimum-weight failure mechanisms. A key component of our framework is the formal notion of \emph{residual errors} and their associated distance metrics, which form lightweight proxies for capturing error propagation and quantifying the potential harm of circuit-level errors. Applying our automated framework to diverse classes of codes, we observe consistent improvements in logical performance of up to an order of magnitude compared to existing single-ancilla SEC designs. 
Moreover, by directly applying these proxies, we prove that no non-interleaving SEC can achieve circuit distance $12$ for the gross code, and we identify an explicit circuit that we conjecture achieves distance $11$, exceeding previously known constructions.
\end{abstract}

\maketitle


\section{Motivation and summary of results} 

As quantum hardware continues to mature, increasing emphasis is being placed on practical quantum error correction (QEC), both theoretically~\cite{yoder2025tourgrossmodularquantum, gidney2024magicstatecultivationgrowing, xu2024constant} and experimentally~\cite{Bluvstein2024LogicalProcessor, GoogleQuantumAI2025BelowThreshold, yamamoto2025quantumerrorcorrectedcomputationmolecular}. 
On the theory side, this effort typically involves three tightly coupled components -- the design of better error-correcting codes~\cite{bravyiEtAl2024, scruby2024highthresholdlowoverheadsingleshotdecodable, malcolm2025computingefficientlyqldpccodes}, the construction of efficient syndrome extraction circuits (SECs)~\cite{kang2025quantum, gidney2023newcircuitsopensource}, and the development of fast and effective decoders~\cite{hillmann2025localizedstatisticsdecodingquantum, muller2025improvedbeliefpropagationsufficient}. 
In this work, we isolate the middle component and introduce a framework for rapidly constructing simple, high-performance SECs for any quantum CSS code~\cite{calderbank1996good,shor1996fault,steane1996multiple}.

We focus exclusively on the common class of \emph{single-ancilla CNOT-based} SECs, as shown in \fig{sec-circuits}.
To develop a framework that applies to arbitrary CSS codes and yields performant SECs, we must address two well-known challenges.
First, achieving low depth requires parallelising the measurement of overlapping checks as much as possible.
This introduces non-trivial \emph{scheduling constraints}, since each qubit can participate in only one CNOT at a time, and interleaving non-commuting CNOTs can invalidate the circuit.

\begin{figure}
\centering
\includegraphics[width=0.95\linewidth]{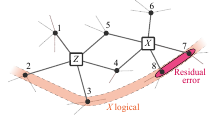}

\vspace{15pt}

\includegraphics[width=0.65\linewidth]{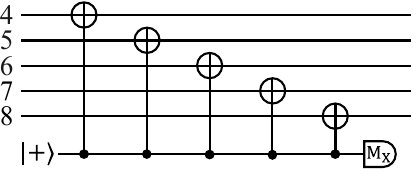}
\caption{
(Above) The Tanner graph of a general quantum CSS code can be highly irregular, with logical operators lacking simple structure. 
A weight-1 error can propagate through the circuit into high-weight residual error (here, qubits $7$ and $8$) that overlaps a minimum-weight logical operator, effectively reducing the code distance.
This is known as a hook error.
(Below) We focus on single-ancilla CNOT-based SECs, which have a dedicated ancilla to measure each check, as shown for a weight-5 X-type check here.
For Z-type checks, the ancilla is prepared in $\ket{0}$ and measured in $Z$ basis, and the CNOT direction is reversed.}
\label{fig:sec-circuits}
\end{figure}

The second challenge in SEC design arises from \emph{hook errors}, whereby weight-1 errors during the circuit can propagate through the CNOT network and produce high-weight \emph{residual errors} on the data qubits (see \fig{sec-circuits}).
These can affect logical performance in several ways~\cite{beverland2024fault}; for example, if the residual error lies within a minimal logical operator of the code, it reduces the \emph{circuit distance}, which can be thought of as the effective distance of the code when implemented by the SEC.
Hook errors can be mitigated using Shor-type syndrome extraction~\cite{shor1996fault}, by augmenting single-ancilla SECs with flag qubits~\cite{chao2018quantum, lingling2020faulttolerant, chamberland2018flag}, or through related gadget constructions~\cite{Derks2025designingfault}.
However, such techniques often incur additional circuit depth and/or qubits, typically worsening the circuit's performance in relevant noise regimes.
A different approach is to choose a CNOT schedule for which the hook errors do not propagate onto data qubits in the support of low-weight logical operators. 
This is well studied in the context of rotated surface codes, where inappropriate CNOT orderings lead to a reduced circuit distance, below the nominal code distance~\cite{Tomita2014lowdistance, Litinski2018latticesurgery}. 
While carefully chosen CNOT schedules can mitigate distance-reducing hook errors in these structured code families, identifying such schedules for general CSS codes is considerably challenging, particularly when low circuit depth is prioritised.

\begin{figure*}
    \centering
    \includegraphics[width=\linewidth]{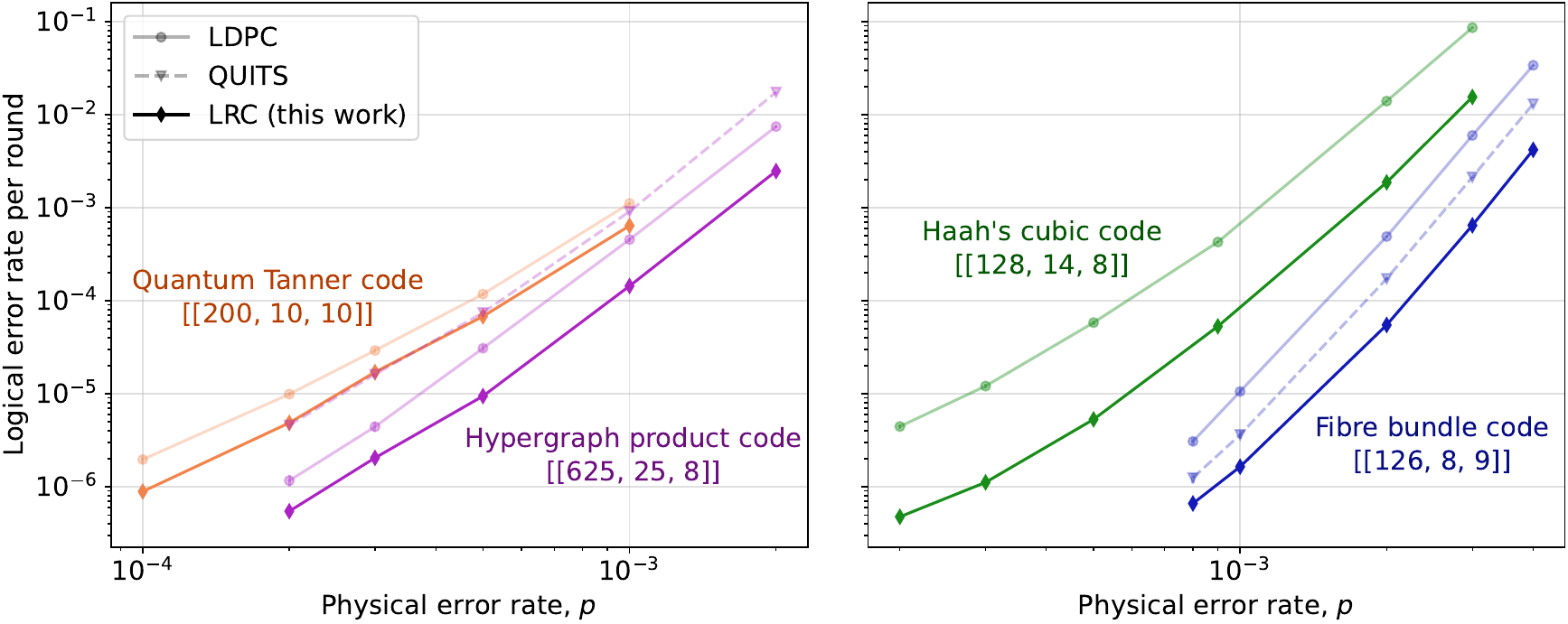}
    \caption{Performance of circuits designed using our left--right circuit (LRC) framework and general QEC packages \texttt{QUITS}~\cite{kang2025quantum} and \texttt{LDPC}~\cite{Roffe_LDPC_Python_tools_2022}, where applicable. 
    We compare SECs of a quantum Tanner code $[[200, 10, 10]]$~\cite{radebold2025explicit}, a hypergraph product code $[[625, 25, 8]]$~\cite{TremblayEtAl2022, kang2025quantum}, a Haah's cubic code $[[128, 14, 8]]$~\cite{haah2011fractal} and a novel fiber bundle code $[[126, 8, 9]]$ that we construct following methods in~\cite{hastings2021fiberbundlecodes}. Most of the error bars are comparable in size to the markers, and therefore, not visible. 
    A Python codebase for generating high-performance left--right circuits for arbitrary CSS codes is publicly available on \href{https://github.com/PurePhys/LR-circuits}{GitHub}. 
    The exact \texttt{Stim} circuits are available on \href{https://doi.org/10.5281/zenodo.18853600}{Zenodo}.
    }
    \label{fig:data}
\end{figure*}

To help overcome the scheduling challenge, we base our framework on \emph{left--right circuits}, a general design for single-ancilla, CNOT-based SECs that partitions the data qubits into two subsets and staggers $X$ and $Z$-type stabiliser checks with respect to that partition (see \fig{lr-circuit}). 
Crucially, this circuit structure sidesteps the interleaving constraints of non-commuting CNOTs.
Previously proposed for specific codes under the name \emph{pipelined product colouration circuits}~\cite{xu2024constant}, we generalise this construction to arbitrary CSS codes and show that left--right circuits can attain low depth for many practically relevant classes of quantum LDPC codes.
Note that this circuit structure is not limited to single-ancilla CNOT-based circuits and has been used in more complex architectures, such as those optimised for Majorana-based hardware~\cite{GransSamuelsson2024improvedpairwise}.

The left--right circuit structure also provides substantial CNOT scheduling freedom which can be exploited to avoid schedules with problematic hook errors.
To efficiently but rigorously identify such schedules, we study properties of the residual-error set $\mathcal{E}$ associated with a given left--right circuit, where each residual error is the set of data qubits to which a weight-1 hook error propagates.

We capture the effect of residual errors using two new generalised distance definitions.
To quantify the intuition that a residual error $E \in \mathcal{E}$ is more harmful when it overlaps a logical operator, we define its \emph{residual distance} $\Delta(E)$ as the smallest number of additional data-qubit errors required, together with $E$, to form a logical operator.
For a distance-$d$ code, if $\Delta(E) < d$, the underlying hook error implies a reduced circuit distance, however if $\Delta(E) > d$, the underlying hook error is potentially less harmful than a single data-qubit error.

The residual distance captures the effect of individual residual errors, but it can be useful to analyse the effect of a set of residual errors $\mathcal{R} \subseteq\mathcal{E}$ in combination. 
For this, we define an extended version of the code, where each residual error is treated as another single-qubit error by appending a column to the check matrix, for which one can compute the \emph{extended-code distance} $d_{\mathrm{ext}}(\mathcal{R})$.
We then prove that for any left--right circuit with residual error set $\mathcal{E}$, the circuit distance $d_{\mathrm{circ}}$ satisfies
\begin{equation}
d_{\mathrm{circ}} = d_{\mathrm{ext}}(\mathcal{E}) \le d_{\mathrm{ext}}(\mathcal{R}) \le \Delta(E)
\qquad
\forall E\in\mathcal{R} \subseteq\mathcal{E}. \label{eq:bounds}
\end{equation} 
This is a special case of our more general \thm{circ_dist_bounds} applicable to any \emph{non-interleaved} SEC, which requires more nuanced definitions to state in full. In addition, we prove that interleaving $X$ and $Z$-type stabiliser checks can only improve the circuit distance compared to the non-interleaved case.

From a practical perspective, since both $\Delta(E)$ and $d_{\mathrm{ext}}(\mathcal{R})$ are defined in the code-capacity model, they are significantly cheaper to compute or upper bound than $d_{\mathrm{circ}}$ defined over the circuit noise model.
Among other things, these distances allow efficient screening of wide range of circuits by rapidly computing an upper bound for $d_{\mathrm{circ}}$.

We use the extended-code distance to prove bounds on the achievable circuit distance of a broad class of SECs for the gross code~\cite{bravyiEtAl2024}. 
We show that no non-interleaved circuit, including any left--right circuit, can attain circuit distance $12$, and that no uniformly tiled non-interleaved circuit can attain circuit distance $11$. 
Applying our bounds to residual errors from a growing sequence of check sets, we progressively restrict the search space to ultimately identify an explicit non-uniform left--right circuit. 
Supported by extensive testing, we conjecture this circuit achieves the optimal circuit distance $d_{\mathrm{circ}} = 11$, improving upon the standard circuit of Ref.~\cite{bravyiEtAl2024} while matching its depth when amortised over many rounds.

Our main contribution is a general framework, together with associated open-source codebase, that automatically searches and constructs high-performance left--right circuits for any CSS code.
To do so, the framework incorporates residual errors and their associated distances into a heuristic for rapidly ranking left--right circuits with different CNOT schedules. 
Circuits in which all (or most) checks exhibit large residual distance $\Delta(E)$ are preferred and ranked higher. 
Together with additional criteria, such as low ancilla idling time, this ranking identifies optimised left--right circuits for a diverse set of quantum LDPC codes, with the performance of each circuit plotted in \fig{data}. 
We find these optimised left--right circuits consistently reduce both circuit depth and logical failure rates relative to alternative SEC constructions, with improvements in logical error rates of up to an order of magnitude.

\section{Preliminaries and Notation}

In this section, we introduce concepts and definitions used throughout the paper. 

\subsection{Syndrome extraction circuits}
\label{sec:SEC}

A CSS stabiliser code on $n$ data qubits is defined by binary parity-check matrices $H_X\in\mathbb{F}_2^{m\times n}$ and $H_Z\in\mathbb{F}_2^{m'\times n}$ obeying $H_X H_Z^{\mathsf T}=0$. 
By specifying $(H_X,H_Z)$, we also fix the stabiliser generators of the code that will be measured.
Specifically each row of $H_X$ (resp.\ $H_Z$) defines an $X$-type (resp.\ $Z$-type) stabiliser generators with support given by the columns equal to $1$. 
We refer to the chosen stabiliser generators as \emph{checks} for short.

Note that the matrix $H \in \{H_X, H_Z\}$ has an equivalent representation as a bipartite Tanner graph, which we also write as $H$, with a circular variable node $v_j$ for each column $j$, a square check node $z_i$ for each row $i$ and an edge $(v_j,z_i)$ iff the $(i,j)$ entry of $H$ is 1. 
We also define the tripartite Tanner graph $H_{XZ}$, which has a circular variable node $v_j$ for each column $j$, a square $Z$-check node $z_i$ for each row $i$ of $H_Z$, and a square $X$-check node $x_i$ for each row $i$ of $H_X$. 
There is an edge $(v_j,z_i)$ iff the $(i,j)$ entry of $H_Z$ is 1 and an edge $(v_j,x_i)$ iff the $(i,j)$ entry of $H_X$ is 1. 
The matrix representation of $H_{XZ}$ is $H_X$ vertically stacked above $H_Z$.
Bipartite Tanner graphs $H_X$ and $H_Z$ are recovered from $H_{XZ}$ by deleting the $Z$-check nodes (resp. $X$-check nodes) and their incident edges. 

Let $\delta(\mathcal{G})$ denote the maximum node degree of any graph $\mathcal{G}$.
We also denote the maximum row and column weights of any matrix $M$ as $r(M)$ and $c(M)$ respectively, such that $\delta(M) = \mathrm{max}(r(M), c(M))$.

A \emph{syndrome extraction circuit} (SEC) measures all checks of $H_X$ and $H_Z$, and is typically repeated for multiple rounds.
Throughout this work, we assume that each check is measured using a dedicated ancillary qubit, the same one for each round.
For an $X$ ($Z$) check, the ancilla is prepared in $\ket{+}$ ($\ket{0}$), CNOT gates are applied from ancilla to data (data to ancilla) qubits on its support and the ancilla is measured in the $X$ ($Z$) basis. 
All circuit operations (preparations, measurements, one- and two-qubit gates, and idle operations) take one time step.
The SEC is specified by assigning, for each check, the time steps of the ancilla preparation, associated CNOTs, and measurement. 
If multiple rounds of an SEC are implemented, they may be staggered so that a new round can begin before the previous one ends. 

In general, CNOTs from checks that share data qubits may be interleaved. 
In contrast, \emph{non-interleaved} circuits are those in which either all $X$-type CNOTs precede all overlapping $Z$-type CNOTs, or all $Z$-type CNOTs precede all overlapping $X$-type CNOTs.
An extreme example of non-interleaved circuits is an \emph{alternating-XZ} circuit, in which all $X$-type checks are measured completely, and then all $Z$-type checks are measured completely before repeating.

\subsection{Depth bounds}

We define the depth of an SEC as
\begin{equation}\label{eq:SEC-depth}
    t=\lim_{m\to\infty} \frac{T^{(m)}}{m},
\end{equation}
where $T^{(m)}$ is the number of time steps required to perform $m$ consecutive rounds of syndrome extraction. 
This definition permits staggered execution of different check measurements, so that operations from different rounds on disjoint qubits can overlap in time and the depth is amortised over many rounds.
We can bound the depth,
\begin{equation}
t \ge \text{max}(c(H_{XZ}),r(H_{XZ}) + 2), \label{eq:depth-bound}
\end{equation}
since some data qubit participates in $c(H_{XZ})$ CNOTs, and some ancilla qubit participates in $r(H_{XZ})$ CNOTs along with an ancilla preparation and a measurement.

\subsection{Noise model and logical error rate estimation}
Unless stated otherwise, we assume standard circuit noise, where each circuit operation fails independently with probability $p$. 
Upon failure, (i) an idle operation is followed by a uniformly random Pauli gate, $X$, $Y$, or $Z$; (ii) a CNOT is followed by a uniformly random non-identity two-qubit Pauli gate; (iii) a preparation $|0\rangle/|+\rangle$ flips to $|1\rangle/|-\rangle$; (iv) a $Z$/$X$ measurement has its outcome flipped.

To characterise the performance of a given SEC for a specific CSS code, we use a fairly standard approach to estimate the logical error probabilities $P^{(m)}_X(p)$ and $P^{(m)}_Z(p)$ for both $X$- and $Z$-basis $m$-round memory experiments with a specified decoder.
These consist of $m$ noisy rounds of the SEC with error rate $p$, followed by a perfect round of syndrome extraction (which represents measuring out all data qubits).
For each memory experiment circuit, we construct the respective \emph{detector error model} using \texttt{Stim}~\cite{Gidney_2021} and \texttt{LDPC}~\cite{Roffe_LDPC_Python_tools_2022} packages. This model includes a binary matrix $H_{\mathrm{circ}}$, where each row $i$ of $H_{\mathrm{circ}}$ corresponds to a detector and each column $j$ to an individual \emph{error mechanism} in the circuit.
The entry $H_{\mathrm{circ}}(i,j)=1$ iff the error mechanism $j$ flips the value of the detector $i$. 
We call a collection of error mechanisms an \emph{error configuration}, which is represented as a bitstring $e$ and produces syndrome $\sigma = H_{\mathrm{circ}}e$. 
In addition to the matrix $H_{\mathrm{circ}}$, the detector error model also specifies the probability distribution for error configurations and the action of each error mechanism on the logical observables which differ between the $X$-basis and $Z$ basis memory experiments (see~\cite{Derks2025designingfault} for a more thorough introduction to detector error models).

To estimate $P^{(m)}_X(p)$ and $P^{(m)}_Z(p)$ in practice, we run many independent trials, each of which samples an error configuration $e$ of the corresponding detector error model and outputs its syndrome $\sigma$. Then, the decoder provides a correction $\hat{e}=\mathcal{D}(\sigma)$.  
If $\hat{e}$ and $e$ have both the same syndrome and the same action on the logical observables, the trial succeeds, and it fails otherwise.
For a distance-$d$ code, we set $m=d$ and calculate the total logical error rate as $P^{(d)} = P^{(m)}_X(p) + P^{(m)}_Z(p)$. Finally, we report the per-round \emph{logical error rate} $P(p):=P^{(d)}(p)/d$.

We take the standard definition of $X$-type and $Z$-type \emph{circuit distance} $d^X_{\mathrm{circ}}$ and $d^Z_{\mathrm{circ}}$ as the minimum number of error mechanisms 
in an $X$-basis or a $Z$-basis memory experiment respectively that flip no detectors but flip at least one logical observable.
We then define the overall \emph{circuit distance} to be $d_{\mathrm{circ}} = \mathrm{min}(d^X_{\mathrm{circ}}, d^Z_{\mathrm{circ}})$.
In the small-$p$ limit and assuming an optimal decoder, the circuit distance determines the leading-order scaling of the logical failure probability
\begin{equation}\nonumber
P(p) \;\xrightarrow[p\to 0]{}\; N_{\mathrm{fail}}\,p^{\lceil d_{\mathrm{circ}}/2\rceil},
\end{equation}
where $N_{\mathrm{fail}}$ denotes the \emph{minimum-weight failure multiplicity}~\cite{fowler2012surface}.
Desirable SECs attain large $d_{\text{circ}}$ and small $N_{\text{fail}}$, and have low depth—both to reduce the number of idling operations (often resulting in lower $N_{\text{fail}}$) and to shorten the logical clock cycle.

Note that, when estimating logical error rates in practice (as in \fig{data}), efficient decoders are typically preferred over optimal ones, since optimal decoding is computationally intractable in general.
We therefore use the belief-propagation ordered-statistics (BPOSD) decoder~\cite{roffe2020} for all codes considered in this work (see \app{sim-details} for parameters).

\section{Left--right circuits}
\label{sec:LRC}

\begin{figure}[t]
  \centering
  \includegraphics[width=1.0\linewidth]{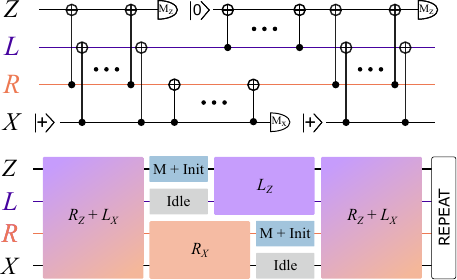}
  \caption{Left–right syndrome-extraction circuit. Data qubits are partitioned into left/right sets, inducing $H_X=[L_X\,|\,R_X]$ and $H_Z=[L_Z\,|\,R_Z]$. CNOTs for $L_X$ and $R_Z$ are executed in parallel for $t_1$ time steps, then $L_Z$ and $R_X$ for $t_2$ time steps. Ancilla preparations/measurements are staggered so all $Z$-check CNOTs precede all $X$-check CNOTs of the same round, ensuring validity and leading to the total circuit depth $t_1+t_2+2$.
  }
  \label{fig:lr-circuit}
\end{figure}

In the following, we describe \emph{left--right circuits} (LRCs), a class of SECs that achieve low depth for a wide range of practically relevant quantum LDPC codes. 
While originally proposed for specific classes of codes~\cite{xu2024constant}, we generalise their application to arbitrary CSS codes.

\subsection{Circuit structure}

To construct an LRC of a CSS code $(H_X, H_Z)$, first partition the $n$ data qubits into two subsets: 
a left set of $l$ qubits and a right set of $r = n-l$ qubits, and write
\begin{align}
H_X &= [ L_X \mid R_X ], \label{eq:pcm1}\\
H_Z &= [ L_Z \mid R_Z ], \label{eq:pcm2}
\end{align}
where $L_X \in \mathbb{F}_2^{m \times l}$, $R_X \in \mathbb{F}_2^{m \times r}$, $L_Z \in \mathbb{F}_2^{m' \times l}$, and $R_Z \in \mathbb{F}_2^{m' \times r}$.  
For each $H\in\{L_X,R_X,L_Z,R_Z\}$, fix a proper edge-colouring $C(H)$ of the bipartite Tanner graph $H$, i.e., a map from edges to $\{1,\dots,|C(H)|\}$ such that edges incident to a common vertex receive distinct colours. 
Here, $|C(H)|$ is the total number of colours used.
Since $H$ is bipartite, there exist efficiently-computable minimal edge-colourings with $|C(H)|=\delta(H)$~\cite{alon2003simple}, although we do not require that the colourings are minimal.

The left--right circuit structure is fully specified by the edge-colourings $C(L_X), C(R_X), C(L_Z),C(R_Z)$ as follows.
CNOTs corresponding to $L_X$ and $R_Z$ are applied in parallel, scheduled by their edge colours (note that we use integers rather than actual colours to label the colouring for this reason).
This requires
\[
t_1 = \max\bigl(|C(L_X)|, |C(R_Z)|\bigr)
\]
time steps (see \fig{lr-circuit}). 
Next, CNOTs corresponding to $L_Z$ and $R_X$ are applied analogously, where the time step for each colour label is shifted by $t_1+1$, requiring 
\[
t_2 = \max\bigl(|C(L_Z)|, |C(R_X)|\bigr)
\]
additional time steps.  
The ancilla initialisations and measurements are staggered in the following way.
The $X$ ancilla is prepared at time step $1$ and measured at time step $t_1 + t_2 + 2$,  
while the $Z$ ancilla is prepared at time step $t_1 + 3$ and measured at time step $t_1 + 2$. Therefore, the total SEC depth (recall \eq{SEC-depth}) is
\begin{equation}
t = t_1 + t_2 + 2.
\label{eq:left--right-depth}
\end{equation}
The staggering guarantees that the circuit is valid, since all CNOTs for the $Z$ checks are applied before the CNOTs for the $X$ checks of the same round, i.e., the SEC is non-interleaved. Note that the circuit design is not limited to a single-time step ancilla initialisation and measurement; longer operations are also accommodated by the same design rules.
Finally, we will refer to a left--right circuit as $\mathcal{C}$, by which we mean the set of four colourings which fully specify the circuit.

\subsection{Low depth LRCs}

While the left--right structure can be applied to any CSS code, it is desirable to choose a qubit partition and edge-colourings that minimise the overall circuit depth.
Many practically relevant code classes have a natural partitioning choice that comes directly from their structure. When combined with optimal edge-colourings, the depth of LRCs with such partitions is often close to optimal.
Here, in the main text, we prove this claim for the hypergraph product (HGP) codes~\cite{tillich2014quantum} as an illustrative example. 
We extend these arguments to a more general class of lifted product (LP) codes~\cite{panteleev2020quantum} (including HGP codes and bivariate bicycle codes as special cases~\cite{bravyiEtAl2024}) in \app{code-partitions}.

A general HGP code has check matrices
\begin{align*}
H_X & = [A \otimes I_{m_B} \,|\, I_{m_A} \otimes B],\\
H_Z & = [I_{n_A} \otimes B^T \,|\, A^T \otimes I_{n_B}],
\end{align*}
for any choice of binary matrices $A \in \mathbb{F}_2^{m_A \times n_A}$ and $B \in \mathbb{F}_2^{m_B \times n_B}$. Here, $I_i$ is identity matrix of size $i$.

\begin{proposition}[HGP circuit depth]
    For any HGP code defined by binary matrices $A$ and $B$, the depth of its minimal colouring LRC is $t = \delta(A) + \delta(B) +2$, close to the lower bound of $\delta(A) + \delta(B)$.
\end{proposition}

Moreover, it is always possible to efficiently identify a minimal colouring using the algorithm in Ref.~\cite{alon2003simple}.

\begin{proof}
Notice that $A \otimes I_{m_B}$ and $I_{n_A} \otimes B^T$ have the same number of columns, allowing a left--right partition forming blocks $L_X = A \otimes I_{m_B}$, $L_Z=I_{n_A} \otimes B^T$, $R_X = I_{m_A} \otimes B$, $R_Z=A^T \otimes I_{n_B}$. 
The bipartite Tanner graph representation of $A \otimes I_s$ is $s$ disconnected copies of $A$, implying $\delta(A \otimes I_s) = \delta(A^T \otimes I_s) = \delta(A)$ and similarly for $B$, where we also used the property $\delta(A^T) = \delta(A)$. 
Therefore, following \eq{left--right-depth}, choosing a minimal colouring $\mathcal{C}(H)$ for each $H \in \{L_X,L_Z,R_X,R_Z\}$ results in an LRC of depth $t = \delta(A) + \delta(B) +2$.
To complete the proof, it is then sufficient to make use of \eq{depth-bound} and show that the degree of the tripartite Tanner graph is $\delta(H_{XZ}) = \delta(A) + \delta(B)$, which we show as follows.

Denote the maximum row and column weights of any matrix $M$ as $r(M)$ and $c(M)$ respectively, such that $\delta(M) = \mathrm{max}(r(M), c(M))$.
First, consider the matrix $A \otimes I_{r_B}$ and partition its rows into blocks of $r_B$ rows. 
Due to the tensor product with identity on the right, there exists a row-block where all of its $r_B$ rows have weight $r(A)$. 
Now, consider the same blocking of rows but instead for matrix $I_{r_A} \otimes B$. 
Due to the tensor product with identity on the left, each row-block has at least one row with row weight $r(B)$. 
Combining both results, there must be a row of $H_X$ that has row weight $r(H_X) = r(A) + r(B)$.
Following similar arguments, the max row weight of $H_Z$ is $r(H_Z) = r(A^T) + r(B^T) = c(A) + c(B)$.
Therefore $r(H_{XZ}) = \mathrm{max}(r(A) + r(B),c(A) + c(B))$.

Moreover, by essentially the same argument, one can show that the maximum column weight of the stacked matrix restricted to the left and right qubits, denoted $H_{XZ}|_L$, satisfies $c(H_{XZ}|_L) = c(A) + r(B)$ and similarly that $c(H_{XZ}|_R) = r(A) + c(B)$.
Therefore, $c(H_{XZ}) = \mathrm{max}(c(A) + r(B),r(A) + c(B))$.

Putting these pieces together, we have  
\begin{eqnarray*}
\delta(H_{XZ}) &=& \mathrm{max}(r(H_{XZ}),c(H_{XZ})), \\
&=& \mathrm{max}(r(A) + r(B), c(A) + c(B),\\ 
&&c(A) + r(B), r(A) + c(B)), \\
&=& \mathrm{max}(r(A), c(A)) + \mathrm{max}(r(B), c(B)),
\\
&=& \delta(A) + \delta(B).
\end{eqnarray*}
\end{proof}
In fact, when $r(H_{XZ}) \geq c(H_{XZ})$, the minimal colouring LRC saturates the bound \eq{depth-bound}, and therefore has optimal depth (rather than just being optimal to within an additive constant 2).
We also refer the reader to a recent proof of the depth of staggered circuits for a special class of cyclic HGP codes in~\cite{aydin2025cyclichypergraphproductcode}.

As shown above, the depth of a left–right circuit is determined solely by the colour counts $|C(L_X)|,|C(R_X)|,|C(L_Z)|,|C(R_Z)|$. 
However, many distinct edge-colourings realise the same counts, producing different SECs with identical depth. 
These choices can significantly alter hook-error propagation, and therefore, the circuit distance $d_{\text{circ}}$ and min-weight failure multiplicity $N_{\text{fail}}$. 
We next introduce concepts that will constitute rapid evaluation and ranking of such circuits.

\section{Residual errors and circuit distance bounds}
\label{sec:Freedom}

In this section, we introduce analytical methods to characterise the effect of hook errors in syndrome extraction circuits. We use these results to rapidly bound circuit distances as part of our SEC design framework, and additionally apply them to the gross code to construct a circuit with distance that we believe to be optimal.

\subsection{Residual errors}

We adopt a simplified model to capture the effect of hook errors from a single check in an LRC. 
We focus on $Z$-checks, however, the $X$-check case is analogous with $Z\!\leftrightarrow\! X$. 

Let $i_1,\ldots,i_w$ be the $w$ data qubits in the support of some $j^{th}$ $Z$-check such that they are listed in the circuit's CNOT order. 
Furthermore, let $e_k\in\mathbb{F}_2^n$ be the unit vector with a $1$ at position $k$ labelling a single data qubit.
We define the circuit's set of \emph{residual errors} $E$ associated  with this check to be
\[
\mathcal{E}_{Z,j}
:= \bigl\{\,E_l=\textstyle\sum_{s=l}^w e_{i_s}\;:\; l=2,\ldots,w\,\bigr\}\subset\mathbb{F}_2^n,
\]
and by $\mathcal{E}_Z = \bigsqcup \mathcal{E}_{Z,j}$ we denote the disjoint union of residual errors from all $Z$ checks. We use the disjoint union to avoid confusing the same residual errors from different checks.
Conceptually, residual errors are Z-type errors on the data qubits which can arise from a single hook error on the $Z$-check ancilla (see \fig{circuits}). 
Note the correspondence between the data qubit index sets and residual errors -- we use them interchangeably.

\begin{figure}[tbp]
    \centering
    \includegraphics[width = 0.9\linewidth]{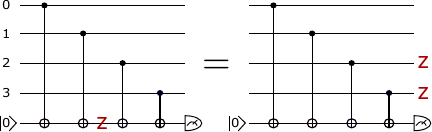}
    \caption{
    Each operation in the circuit involves at most one data qubit, yet an error can propagate to multiple data qubits through subsequent steps. 
    For instance, a $Z$-error on the ancilla after a CNOT gate propagates to every data qubit involved in later CNOT operations. 
    For the given check, the set of all residual errors is $\{\{3\}, \{2,3\}, \{1, 2, 3\}\}$. 
    Here, we display the residual error $\{2,3\}$.
    }
    \label{fig:circuits}
\end{figure}

A residual error is generally more harmful if it significantly overlaps a non-trivial logical operator, in some cases, even a single qubit overlap can make a noticeable difference. 
To quantify this, we introduce two distance-related metrics applied to either a residual error or a set of residual errors. 

\begin{definition}[Residual distance]\label{def:residual_dist}
For a CSS code with parity check matrices $H_X$, $H_Z$, the residual distance of a $Z$-type residual error $E\in\mathcal{E}_Z$ is given as
\begin{equation}
\Delta(E):=1\!+\!\min_{D\in\mathbb{F}_2^n}\bigl\{|D| \! :\ \!\! (E\!+\!D) \in \ker (H_X) \!\setminus\! \operatorname{span}(H_Z)\bigr\}, \nonumber
\end{equation}
where $|\cdot|$ is the Hamming weight.
\end{definition}
That is, the residual distance $\Delta(E)$ is the smallest number of data qubit errors (in addition to the residual error $E$) that are required to complete any non-trivial logical operator. Note that $\Delta(E)$ can also exceed code distance $d$, indicating that some residual errors may be less harmful than weight-1 errors. Consequently, $\Delta(E)$ provides a quantitative measure of the potential harm caused by $E$.

To define the second distance-related metric, we need to first describe the concept of the \emph{extended code}, which builds on the extended Tanner graph introduced in ~\cite{pacenti2025turboannihilationhookerrorsstabilizer}. 
An extended code is an extension of the CSS code, where additional columns, reflecting the hook errors, are added to the corresponding check matrix for each residual error considered. 
The extra columns are similarly applied to the logical operator matrices $L_X$, $L_Z$.

\begin{definition}[Extended code]\label{def:ext_code}
Given $H_X$, $L_X$ and a set of $Z$-type residual errors $\mathcal{R}\subseteq \mathcal{E}_Z$ for an SEC, define the extended matrices $H_X^{\mathrm{ext}}(\mathcal{R})$, $L_X^{\mathrm{ext}}(\mathcal{R})$ by appending one column to $H_X$ and one column to $L_X$ for each residual error in $\mathcal{R}$.
A residual column is defined by adding the columns of $H_X$, $L_X$ (mod 2) associated with each data qubit in the residual error.
\end{definition}

\noindent For example, if a given residual error has support on data qubits $\{i, j\}$, then the appended column to $H_X^{\mathrm{ext}}$ corresponding to this residual is equal to the sum of $i$ and $j$ columns (mod 2) of $H_X$.
With this, we can define the \emph{extended-code distance}.

\begin{definition}[Extended-code distance]\label{def:ext_code_dist}
The $Z$-type extended-code distance $d^Z_{\mathrm{ext}}(\mathcal{R})$ of a set of residual errors $\mathcal{R}\subseteq \mathcal{E}_Z$ is the minimum Hamming weight of a bit string $x$, such that
\begin{equation}\label{eq:ext_code_dist}
H_X^{\mathrm{ext}} x^{\mathsf T} = 0
\quad\text{and}\quad
L_X^{\mathrm{ext}} x^{\mathsf T} \neq 0 .
\end{equation}
\end{definition}
\noindent 
We will consider both the full residual set $\mathcal{R} = \mathcal{E}_Z$ and subsets $\mathcal{R}$ corresponding to the residuals for individual checks. 
The corresponding definitions for $X$-type residual errors and extended-code distance are obtained analogously by exchanging $X$ and $Z$. Finally, both types of residual error sets can be combined using the disjoint union $\mathcal{R} = \mathcal{R}_X\sqcup \mathcal{R}_Z$, which allows us to distinguish the residual errors of each type.
Then, denote by $d_{\mathrm{ext}}(\mathcal{R})$ the combined extended code distance
\begin{equation}\label{eq:ext_dist_combined}
    d_\mathrm{ext}(\mathcal{R}) = \min (d^X_\mathrm{ext}(\mathcal{R}_X), d^Z_\mathrm{ext}(\mathcal{R}_Z)),
\end{equation}
where $\mathcal{R}_X \subseteq \mathcal{E}_X$ and $\mathcal{R}_Z \subseteq \mathcal{E}_Z$ correspond to subsets of the $X$ and $Z$-type residual errors. When $\mathcal{E} = \mathcal{E}_X \sqcup \mathcal{E}_Z$, we will refer to $d_\mathrm{ext}(\mathcal{E})$ as the \emph{full extended-code distance}.

While the residual distance and the extended-code distance are closely related metrics, they differ in how they modify the underlying distance computation. The former can be viewed as altering the constraints in the code-distance search, whereas the latter introduces additional variables into the model. This distinction, together with the different scales at which the two metrics can operate, leads to specialised use cases for each.
Because both the residual distance and the extended-code distance are based on the code capacity model, they can be estimated on significantly shorter timescales than the full circuit-level distance $d_{\mathrm{circ}}$. 
In fact, computing the residual distance of a given error $E$ is computationally equivalent in size to estimating the original code distance.

\subsection{Circuit distance bounds}

In this subsection, we establish the following provable bounds relating the residual distance and the extended-code distance to the circuit distance of syndrome extraction circuits. 

\begin{theorem}[Circuit distance bounds]
\label{thm:circ_dist_bounds}
Let $\mathcal{C}$ be any single-ancilla, CNOT-based syndrome extraction circuit for a CSS code, and let $\mathcal{C}'$ be a non-interleaved circuit obtained from $\mathcal{C}$ by preserving the relative CNOT ordering within each check. 
Denote their circuit distances by $d_{\mathrm{circ}}$ and $d'_{\mathrm{circ}}$, and let $\mathcal{E} = \mathcal{E}_X \sqcup \mathcal{E}_Z$ be their complete set of residual errors (the same for both $\mathcal{C}$ and $\mathcal{C}'$). 
Then
\begin{align}
d'_{\mathrm{circ}} &\le  \Delta(E) \quad \forall E \in \mathcal{E}, \\
d'_{\mathrm{circ}} &\le d_{\mathrm{ext}}(\mathcal{R}) \quad \forall \mathcal{R} \subseteq \mathcal{E}, \\
d'_{\mathrm{circ}} &= d_{\mathrm{ext}}(\mathcal{E}), \\
d'_{\mathrm{circ}} &\le d_{\mathrm{circ}}.
\end{align}
\end{theorem}

Before proving the statements in \thm{circ_dist_bounds}, we highlight some of their implications. 
The first and second inequalities enable efficiently-computable upper bounds on $d'_{\mathrm{circ}}$, allowing non-interleaved circuits with low circuit distance to be ruled out rapidly.
The equality is the statement that the second bound is tight when the full set of residuals is included in the extended code, and is particularly powerful -- it reduces circuit-level distance computation for non-interleaved circuits to a much smaller problem, potentially allowing the exact distance (not just a bound) to be determined using exhaustive numerical methods.
The last inequalities shows that if we have succeeded in identifying a non-interleaved circuit with a high circuit distance (perhaps using the previous bounds), then forming an interleaved version of the circuit cannot reduce the distance but may further improve it.

In the following propositions, we prove each aspect of \thm{circ_dist_bounds}.

\begin{proposition}[Residual distance bound]\label{propos:res_dist_bound}
    The circuit distance $d'_{\mathrm{circ}}$ of any non-interleaved SEC with the residual error set $\mathcal{E}$ satisfies
    \begin{equation*}
        d'_{\mathrm{circ}} \leq \Delta(E) \quad \forall E\in \mathcal{E}.
    \end{equation*}
\end{proposition}

\begin{proof}
Consider an SEC with non-interleaved $X$ and $Z$-type checks, and any $Z$-type residual error $E\in \mathcal{E_\mathrm{Z}}$. Then, by introducing additional error-free idling operations on data qubits in the circuit, the SEC can be temporally partitioned into $X$- and $Z$-type parts, during which the corresponding checks are implemented. Note that for any non-trivial SEC, each data qubit undergoes at least one error-prone operation in each partition.

Within a given $Z$-type partition, there exists a weight-1 hook error of $Z$-type that propagates \emph{solely} onto the data qubits $q\in E$ at the end of the partition. Supplementing this with $\Delta(E) - 1$ data qubit errors of $Z$-type at any point of the partition, one can complete a non-trivial logical operator (see \defn{residual_dist}). Because this occurs entirely within a single partition, no detectors are triggered, and the error configuration therefore constitutes a circuit-level logical error. Since the introduction of additional error-free idling operations cannot change the distance of the circuit, we obtain that $d'_{\mathrm{circ}}$ is at most $\Delta(E)$. The same argument applies for any $X$-type residual error $E\in \mathcal{E_\mathrm{X}}$.
\end{proof}


Next, we prove a more general statement that applies to any single-ancilla CNOT-based SEC and forms the backbone for the remaining part of the proof.
\begin{proposition}[Extended-code distance bound]\label{propos:circ_dist_bound}
    The circuit distance $d_{\mathrm{circ}}$ of any SEC with the residual error set $\mathcal{E}$ satisfies
    \begin{equation*}
        d_{\mathrm{circ}} \geq d_\mathrm{ext}(\mathcal{E}).
    \end{equation*}
\end{proposition}

\begin{figure}[tbp]
    \centering
    \includegraphics[width = 0.9\linewidth]{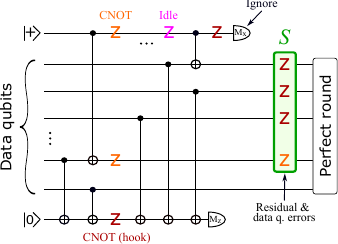}
    \caption{A guide for the proof of \propos{circ_dist_bound}. Regardless of the CNOT network, a non-trivial circuit-level logical operator must have a support $S \in \ker (H_X) \setminus \operatorname{span}(H_Z)$ on data qubits at the end of the circuit. This support only contains singular data qubit errors or multi-qubit residual errors. Note that they may overlap and partially cancel out each other. Finally, note that the indicated hook error (red) is equivalent to a two-qubit $ZZ$ error after the subsequent CNOT. Hence, any $Z$-type CNOT error on data/$Z$-ancilla pair is equivalent to either a single-qubit hook or data qubit error (this was also observed in~\cite{shaw2025lowering}).
    }
    \label{fig:proof_fig}
\end{figure}

\begin{proof}
We present the proof in several steps. For guidance, refer to \fig{proof_fig} throughout.

1. \textit{CSS code.} Since we assume a CSS code, we can consider $X$- and $Z$-type errors, and therefore, $X$-type circuit distance ($d^X_{\mathrm{circ}}$) and $Z$-type circuit distance ($d^Z_{\mathrm{circ}}$) separately. 
To see this, consider pure $X$-type circuit-level errors. These errors have no effect on any $X$-type check or any $X$-type logical observable measured at the end of the circuit, due to the circuit having no gates that map $X \rightarrow Z, Y$ and due to $X$-type errors commuting with all $X$-type observables. Hence, for the purposes of lower bounding $d^Z_{\mathrm{circ}}$ we can ignore all $X$-type errors. Furthermore, this allows us to model $Y$-type errors (up to phase) as $Y = XZ \equiv Z$, and therefore, all $Y$-type errors are already accounted for by only considering $Z$-type errors. The same applies to mixed type two-qubit errors -- any $X$-component of the error can be ignored. In conclusion, one can consider only $Z$-type errors for $X$-type observables (for the purpose of evaluating $d^Z_{\mathrm{circ}}$) and vice-versa. For simplicity, we will only consider $Z$-type errors and $d^Z_{\mathrm{circ}}$, the opposite case with $X$-type errors and $d^X_{\mathrm{circ}}$ follows analogously.

2. \textit{Perfect final round of SEC.} In our setting, the final round of syndrome extraction is assumed to be perfect, which, in practice, manifests as individually measuring data qubits in the basis of interest and then reconstructing the syndrome. Since any circuit-level logical operator is defined as triggering no detectors, any \emph{non-trivial} circuit-level logical error $L_{\mathrm{circ}}$, at the end of the circuit, must have its support on the data qubits $S = \mathrm{supp}_{\mathrm{data}}(L_{\mathrm{circ}})\in\mathbb{F}_2^n$, such that $S \in \ker (H_X) \setminus \operatorname{span}(H_Z)$ (this includes errors from the final data qubit measurements). This means that every circuit-level logical error propagated onto data qubits at the end of the circuit (after data qubit measurements) has an equivalent logical error in the code-capacity model. Therefore, to bound $d^Z_{\mathrm{circ}}$ we must find the lowest number of circuit-level errors that lead to a logical error supported on data qubits at the end of the computation.

3. \textit{Noisy operations.} Let us now consider all error mechanisms that can occur during multiple rounds of SEC, and their resulting support on data qubits after propagating them to the end of the circuit. For convenience, we have collected them in~\tab{example}.

\begin{table}[ht]
\centering
\begin{tabular}{c|c}
\textbf{Error mechanism} & \textbf{\# of data qubit errors} \\ \hline
Init error (data) & 1 \\ 
Init error (any ancilla) & 0 \\ 
Meas error (data) & 1 \\ 
Meas error (any ancilla) & 0 \\ 
CNOT error (Z ancilla/data) & Resp. residual error or 1 \\ 
CNOT error (X ancilla/data) & 0 or 1 \\
Idling error (data) & 1 \\ 
Idling error (X ancilla) & 0 \\
Idling error (Z ancilla) & Resp. residual error \\
\end{tabular}
\label{table:example}
\caption{Types of circuit error mechanisms and their respective support on data qubits at the end of the circuit.}
\end{table}
Note that any $Z$-type error on data qubits can propagate onto an $X$-type ancilla but due to the CNOT directionality, never back onto the data qubits. For the same reason, no $Z$-type error can propagate onto a $Z$-type ancilla at any point of the circuit. Note that this does not depend on whether the checks are interleaved.

4. \textit{Extended-code distance.} We have just shown that single-qubit data errors together with all residual errors capture every potential way to form an error supported on the data qubits at the end of the circuit. This is exactly modelled by the extended-code with $\mathcal{R} = \mathcal{E}_Z$ (see \defn{ext_code}), where each column of $H_X^\mathrm{ext}(\mathcal{E}_Z)$ and $L_X^\mathrm{ext}(\mathcal{E}_Z)$ corresponds to either a single data qubit or a residual error. Therefore, with respect to errors appearing on data qubits, we have a natural mapping from every circuit-level error mechanism to a column of the extended-code model (if the error mechanisms leads to 0 data qubit errors, then it is not mapped to any column). By extension, we can map every circuit-level error configuration $E_{\mathrm{circ}}$ composed of potentially multiple error mechanisms to a bit-string $x$ of the extended-code model.
If at the end of the circuit, a given $E_{\mathrm{circ}}$ has a support $S$ on data qubits, such that $S \in \ker (H_X) \setminus \operatorname{span}(H_Z)$, then its respective bit string $x$ satisfies \eq{ext_code_dist} (by \defn{ext_code_dist}). Hence, $E_{\mathrm{circ}}$ must be composed of at least $d^Z_\mathrm{ext}(\mathcal{E}_Z)$ error mechanisms. Since we already established that for any non-trivial circuit-level logical error $L_{\mathrm{circ}}$ its data qubit support must satisfy $S \in \ker (H_X) \!\setminus\! \operatorname{span}(H_Z)$, we conclude $d^Z_{\mathrm{circ}} \geq d^Z_\mathrm{ext}$. 

Following the same four steps, we can similarly derive that $d^X_{\mathrm{circ}} \geq d^X_\mathrm{ext}$. We conclude the proof by combining distances together using \eq{ext_dist_combined} and $d_\mathrm{circ} = \min (d^X_\mathrm{circ}, d^Z_\mathrm{circ})$.
\end{proof}

The reason we obtain an inequality (and not equality) in \propos{circ_dist_bound} is due to the fact that error propagation happens over multiple time steps. Before the errors have spread onto the data qubits to form an undetectable logical operator, parts of them may have been detected. One way to avoid this detection is to consider non-interleaved circuits.

\begin{corollary}\label{corol:rest_dist}
    The circuit distance $d'_{\mathrm{circ}}$ of any non-interleaved SEC with the residual error set $\mathcal{E}$ satisfies
    \begin{align*}
        d'_{\mathrm{circ}} &\le d_{\mathrm{ext}}(\mathcal{R}) \quad \forall \mathcal{R} \subseteq \mathcal{E}, \\
        d'_{\mathrm{circ}} &= d_{\mathrm{ext}}(\mathcal{E}).
    \end{align*}
\end{corollary}

\begin{proof}
The proof is conceptually similar to that of \propos{res_dist_bound}. Consider an SEC with non-interleaved $X$ and $Z$-type checks. Then, by introducing additional error-free idling operations on data qubits in the circuit, the SEC can be temporally partitioned into $X$- and $Z$-type parts, during which the corresponding checks are implemented.

Within a given $Z$-type partition, every possible $Z$-type data qubit and residual error can appear; and at the end of the partition, the error configuration has support only on data qubits. Therefore, for any weight-$w$ bit string $x$ satisfying~\eq{ext_code_dist} for any residual error subset $\mathcal{R} \subseteq \mathcal{E}$, there exists a weight-$w$ error configuration of only $Z$-type errors within that partition that induces a logical error on data qubits in the code capacity setting, i.e. $L \in \ker (H_X) \setminus \operatorname{span}(H_Z)$. Because this occurs entirely within a single partition, no detectors are triggered, and the error configuration therefore constitutes a circuit-level logical error. Since the introduction of additional error-free idling operations cannot change the distance of the circuit, we obtain that $d'_{\mathrm{circ}}$ is at most $d_{\mathrm{ext}}(\mathcal{R})$. 
An analogous argument applies to $X$-type errors and their corresponding partition.
Finally, combining this result with \propos{res_dist_bound}, we obtain the equality $d'_{\mathrm{circ}} = d_{\mathrm{ext}}(\mathcal{R})$ for the case when $\mathcal{R} =\mathcal{E}$.
\end{proof}

\noindent The final equation of \thm{circ_dist_bounds} follows from combining the results of \propos{res_dist_bound} and \corol{rest_dist}, and with that we finish the proof. 

Recall that LRCs are non-interleaved SECs, and therefore according to \thm{circ_dist_bounds}, for any LRC with the residual error set $\mathcal{E}$ the following holds
\begin{equation}\nonumber
d_{\mathrm{circ}} = d_{\mathrm{ext}}(\mathcal{E}) \le d_{\mathrm{ext}}(\mathcal{R}) \le \Delta(E)
\qquad
\forall E\in\mathcal{R} \subseteq\mathcal{E},
\end{equation}
where we have used that adding additional residual errors for the extended code cannot increase its distance, and that $d_{\mathrm{ext}}(\{E\})\leq \Delta(E)$. This inequality follows from the fact that an extended code of a single variable $d_{\mathrm{ext}}(\{E\})$ includes the case in which the residual error is fixed, that is, it includes the setting of $\Delta(E)$.
As part of this work, we extensively use the above (in)equality chain for a rapid construction of high-performance LRCs.

\subsection{Application to the gross code}

The gross code is a promising example of a relatively small quantum LDPC code~\cite{bravyiEtAl2024,yoder2025tourgrossmodularquantum}, with parameters $[[144,12,12]]$ with weight-six stabiliser checks (see \fig{gross-code} in \app{gross-analysis}). 
Although the code distance is $12$, the best currently-known SECs for the gross code achieve a circuit distance of at most $10$, making it a natural candidate for improved SEC design.

We use \thm{circ_dist_bounds} to prove that for non-interleaved single-ancilla circuits for the gross code:
(i) no circuit achieves circuit distance $12$;
(ii) no uniformly tiled circuit achieves circuit distance $11$; and
(iii) separating checks into 3 colours and tiling a different circuit (each with a specific set of residuals) for each colour leaves very few candidates for circuit distance 11, which are highly constrained.
We use the remaining candidate residual pattern to infer an explicit circuit, which is a slightly modified LRC.
This circuit has depth $8$, which can be shown to be minimal using \eq{depth-bound}, which matches the depth of the original circuit in Ref.~\cite{bravyiEtAl2024}. 
Using extensive testing, we conjecture that it has circuit distance 11.

Here, we provide the formal proposition and proof of (i), deferring the proof and propositions of (ii) and (iii) to \app{gross-analysis}.

\begin{proposition}[Gross circuits]
\label{propos:all-circuits-bound}
Let $d_{\mathrm{circ}}$ be the circuit distance of any non-interleaved SEC for the gross code. Then,
\begin{equation}
d_{\mathrm{circ}} \ \le\ 11.
\end{equation}
\end{proposition}

\begin{proof}
Our strategy is to fix a single check of the code and enumerate all possible sets of residual errors that can arise from that check alone under any CNOT ordering. 
Any non-interleaved SEC must have a residual-error set containing one of these enumerated sets as a subset. 
Therefore, by \thm{circ_dist_bounds}, any extended-code distance bound that holds for all of the enumerated sets also holds for the circuit distance for all non-interleaved SECs.

Each check of the gross code has weight six, and in any non-interleaved SEC each check contributes exactly three residual errors of weight greater than one. 
The resulting residual-error sets depend only on the relative ordering of the CNOTs used to measure the check, yielding exactly $90$ possibilities which can be generated as follows. 
Label the six qubits in the support of a check by $1,\ldots,6$, and suppose the CNOTs are applied in the order $(1,2,3,4,5,6)$. 
The three residual errors contributed to $\mathcal{E}$ are then supported on
\[
\bigl\{ \{1,2\},\ \{4,5,6\},\ \{5,6\} \bigr\}.
\]
(Note that $\{1,2,3,4\}$ is stabiliser-equivalent to $\{5,6\}$, etc.) Permuting the ordering $(1,2,3,4,5,6)$ produces all $90$ inequivalent residual-error sets.

For concreteness, fix a specific $X$ check of the gross code (by symmetry it does not matter which).
Let $\mathcal{R}_i$ for $i=1,\dots,90$ denote the three errors of the residual-error set contributed by that single check for the $i$th of the $90$ possible orderings. 
Any non-interleaved SEC must have a residual-error set $\mathcal{E}$ that contains one of these sets $\mathcal{R}_i$ as a subset. 
Applying the upper bound of \thm{circ_dist_bounds} therefore gives
\[
d_{\mathrm{circ}} \le \max_i d_{\mathrm{ext}}(\mathcal{R}_i).
\]
For each $\mathcal{R}_i$, we upper-bound $d_{\mathrm{ext}}(\mathcal{R}_i)$ by searching for a logical operator in the corresponding extended code, as described in \app{bp-osd-upper-bounds}. 
In every case we find a logical operator of weight at most $11$ (45 of the 90 cases have weight at most $10$), implying $d_{\mathrm{circ}} \le 11$ as required.
\end{proof}

Having shown that circuit distance 12 is unattainable for non-interleaved circuits, in \app{gross-analysis} we use similar techniques, together with additional structure of the gross code, to prove further restrictions and ultimately identify a circuit we conjecture has circuit distance 11. 
The key observation is that by looking at residual error sets of a growing set of checks can allow one to prune the search space rapidly.
For example, in the proof above, 45 of the 90 possible residual-error sets for the check considered already imply circuit distance at most $10$, showing that any SEC containing one of these sets cannot achieve circuit distance $11$.

Finally, while we have ruled out circuit distance $12$ for non-interleaved circuits of the gross code (which include left--right circuits as a special case), our proof does not apply to interleaved circuits. 
It therefore remains an open question whether any interleaved circuit can achieve circuit distance $12$.

\section{Circuit design framework}
\label{sec:framework}

\begin{figure*}[tbp]
    \centering
    \includegraphics[width =0.98\linewidth]{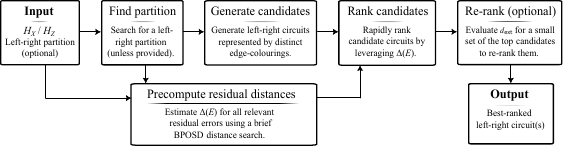}
    \caption{
    The framework workflow for constructing high-performance left-right circuits.
    }
    \label{fig:pipeline}
\end{figure*}

In this section, we detail the workflow of our framework for constructing high-performance left--right circuits. 
We begin with a brief description of each stage, and then describe in detail the methodology used to rank candidate circuits.

\subsection{Framework workflow}

We summarise the framework workflow in \fig{pipeline}. 

\textbf{Input and output.}
The framework requires the parity-check matrices $H_X$ and $H_Z$ of a CSS code, and can optionally be provided with a left--right partition if available (a natural choice often exists for structured code families).
The output of the framework is an explicit optimised left-right circuit for the code (or optionally a ranked list of left-right circuits).

\textbf{Find partition.}
If no partition is provided explicitly, a search-based approach is used in which the columns of $H_X$ and $H_Z$ are simultaneously permuted to generate a set of new left--right partitions $\{L_X, R_X, L_Z, R_Z\}$. 
The degree of each of these Tanner subgraphs sets the number of colours in its minimal colouring, which then sets the depth of all candidate circuits we generate later. 
This allows us to estimate the resulting circuit depth prior to the colouring stage, enabling rapid filtering of unfavourable partitions. 
In practice, well-chosen partitions lead to low-depth circuits.

\textbf{Precompute residual distances.} 
Given the left--right circuit structure, only a limited subset of all potential residual errors need to be considered for a chosen partition. 
For example, if a given check has $3$ left and $3$ right data qubits, it cannot have a residual error consisting of only $2$ left and $2$ right qubits for any LRC.
For computing residual distances, we use the method described in \app{bp-osd-upper-bounds} to perform a rapid and lightweight residual distance estimate.
Since the task of computing $\Delta(E)$ is based on the code-capacity model, for smaller codes, more accurate but computationally more expensive methods (e.g. mixed-integer programming) can be used instead.

\textbf{Generate candidates.}
Recall that left--right circuits are specified by edge colourings of the Tanner subgraphs $L_X, R_X, L_Z,$ and $R_Z$, where each colour specifies the time step of the corresponding CNOT (see \sect{LRC}). 
A minimal edge colouring for each graph can be constructed efficiently using an algorithm based on Hopcroft--Karp~\cite{hopcroft1973matching}, although in practice the runtime is not negligible.
This produces a left-right circuit with the lowest possible depth given the specified partition.
One approach to generating many candidate circuits is to introduce randomness into the colouring algorithm and run it multiple times. 
Instead, in our framework we construct a single colouring and use it to generate additional candidates by permuting the colours within each subgraph, which does not change the circuit depth. 
For many quantum LDPC codes the resulting permutation space is manageable and all permutations can be considered. 
However, for codes with large qubit or stabiliser degrees it may be necessary to restrict the number of permutations explored, as this space grows factorially with the number of colours used.

\textbf{Ranking and optional re-ranking.}
The candidate circuits are ranked as described in the next subsection, forming the output of the framework. 
The initial ranking leverages the precomputed residual distances and uses only metrics which are very fast to evaluate for each circuit.
An optional re-ranking makes use of a more expensive metric which can be applied to a small subset of promising candidate circuits.

\subsection{Ranking candidate circuits}

Here we first state several metrics that can be evaluated for any SEC $\mathcal{C}$ and then describe how they are used to rank candidate circuits in our framework.

For a given circuit $\mathcal{C}$, the residual error set $\mathcal{E}(\mathcal{C}) = \mathcal{E}_X(\mathcal{C}) \sqcup \mathcal{E}_Z(\mathcal{C})$ is easily determined, and the residual distance $\Delta(E)$ for each $E \in \mathcal{E}(\mathcal{C})$ is obtained by lookup, since estimates for all possible residual errors have been precomputed.

\textbf{Min residual distance.}
The simplest metric is the smallest residual distance present in the circuit,
\[
\Delta_\text{min}(\mathcal{C}) := \min_{E\in\mathcal{E}(\mathcal{C})} \Delta(E).
\]
This captures the effect of the most harmful residual error when considered individually and is very fast to compute using the precomputed residual distances.

\textbf{Residual distance profile.}
Using the same precomputed distances, a more detailed profile of residual distances is
\[
W(\mathcal{C}) := (\gamma_1, \gamma_2, \dots),
\]
where $\gamma_u$ counts the number of residual errors with residual distance $u$. 
The index of the first nonzero entry of $W(\mathcal{C})$ is therefore $\Delta_\text{min}(\mathcal{C})$, while the magnitude of this entry indicates how many residual errors attain this limiting distance. 
Subsequent entries record the multiplicities of larger residual distances, and therefore, capture additional information about potentially less harmful residual errors. 
In practice, the vector $W(\mathcal{C})$ can be truncated after a small number of non-zero entries, since only the residuals with the lowest few residual distances significantly influence the circuit performance.

\textbf{Extended-code distance.}
The extended-code distance $d_\text{ext}(\mathcal{C})$ captures interactions between multiple residual errors which are missed by the residual-distance metrics.
However, it requires estimating the extended-code distance for each circuit (see \app{bp-osd-upper-bounds}) without exploiting the precomputed residual distances, making it much slower to evaluate than $\Delta_\text{min}(\mathcal{C})$ or $W(\mathcal{C})$.

\textbf{Ancilla idle count.}
This metric $\tau_A(\mathcal{C})$ is defined as the total number of idle operations summed over all ancilla qubits and all time steps in a round of the SEC $\mathcal{C}$.
It can be computed quickly given the circuit.

Reducing $\tau_A(\mathcal{C})$ is preferred as it lowers the number of potential error locations. 
Data-qubit idling is not included because, once the checks are fixed, the circuit depth determines the total number of data-qubit idle operations, and all candidates in our framework have the same depth.
In contrast, the ancilla idle count depends on the ordering of CNOT gates within each check, since preparation can be delayed until the first CNOT and measurement advanced to immediately follow the last CNOT on the ancilla, thereby reducing $\tau_A(\mathcal{C})$ (see \fig{ancilla_idling}).
The variation in $\tau_A(\mathcal{C})$ is particularly relevant for codes with heterogeneous check weights, or when the maximum check weight is smaller than the maximum data-qubit degree, i.e., when $r(H_{XZ}) < c(H_{XZ})$.

\begin{figure}[tbp]
    \centering
    \includegraphics[width = 1\linewidth]{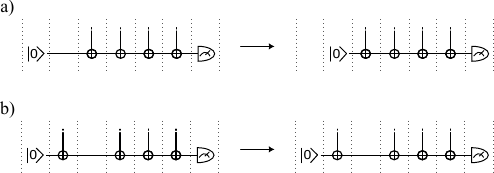}
    \caption{(a) Idle operations immediately after ancilla preparation or just before measurement can be removed by delaying preparation or advancing measurement, reducing the idle count. 
    (b) Idle operations between CNOT gates generally cannot be removed without modifying the circuit.}
    \label{fig:ancilla_idling}
\end{figure}

\textbf{Ranking rules.}
Using the above metrics, candidate LRCs arising from different edge colourings of $L_X$, $R_X$, $L_Z$, and $R_Z$ are ranked according to the following criteria, applied in order:
\begin{enumerate}
    \item \textbf{Min residual distance, $\Delta_\text{min}$:} circuits with larger $\Delta_\text{min}$ are ranked higher.
    \item \textbf{Ancilla idle count, $\tau_A$:} if $\Delta_\text{min}$ is equal, circuits with smaller $\tau_A$ are ranked higher.
    \item \textbf{Residual distance profile, $W$:} if both previous values are equal, circuits whose residual distance profile is lexicographically smaller are ranked higher.
\end{enumerate}
If all three criteria coincide, ties are broken using a fixed ordering of the candidate circuits. 
This defines a total ordering over the set of circuits being compared.

These rules reflect the following considerations. 
The circuit distance $d_{\mathrm{circ}}$ is upper bounded by the minimum residual distance, $d_{\mathrm{circ}} \le \Delta_\text{min}(\mathcal{C})$, so maximising $\Delta_\text{min}(\mathcal{C})$ favours circuits with larger potential distance. 
Minimising $\tau_A(\mathcal{C})$ reduces the number of idle locations and therefore the exposure of ancilla qubits to noise. 
Finally, the residual distance profile $W(\mathcal{C})$ captures the multiplicity of residual errors at each distance, and prioritising smaller profiles reduces the number of locations where harmful hook errors may arise. 
Although circuit-level error propagation is more complex than this simplified picture, these metrics provide an effective heuristic in practice (see \app{Residual_analysis}), which can be evaluated inexpensively.

For the top-ranked candidates, the extended-code distance $d_{\mathrm{ext}}$ can be optionally computed to estimate the circuit distance (see \thm{circ_dist_bounds}). 
Circuits with larger estimated $d_{\mathrm{ext}}$ can then be promoted in the ranking.

\section{Example applications}
\label{sec:numerics}

\begin{table*}[tbp]
\centering

\setlength{\tabcolsep}{6pt}
\renewcommand{\arraystretch}{1.2}
\begin{tabular}{lcccccc}
\hline
\textbf{Label} & \textbf{Code class} & \textbf{Parameters} & \textbf{Max degree} & \textbf{Natural L/R} & \textbf{LRC depth} & \textbf{Reference} \\
\hline
HGP625 & Hypergraph product & $[[625,25,8]]$   & 8 & Y &  10 &   \cite{TremblayEtAl2022,kang2025quantum} \\
Tanner200 & Quantum Tanner & $[[200,10,10]]$  & 12 & N & 18 & \cite{radebold2025explicit} \\
Haah128 & Haah's cubic & $[[128,14,8]]$   & 8 & Y & 10 & \cite{haah2011fractal} \\
FB126 & Fibre bundle & $[[126,8,9]]$    & 6 & Y & 8 & \cite{hastings2021fiberbundlecodes}* \\
\hline
\end{tabular}
\caption{Example codes considered under our framework. We also indicate whether the code has a natural left--right partition in the column `Natural L/R'. FB126 is constructed following the methods in \cite{hastings2021fiberbundlecodes}, however, to our knowledge, the code is novel.
}
\label{table:qec_codes}
\end{table*}

To demonstrate our circuit design framework in practice, we apply it to diverse classes of quantum LDPC codes. Specifically, we consider a hypergraph product (HGP) code with parameters \([[625,25,8]]\), constructed in~\cite{TremblayEtAl2022,kang2025quantum}; a quantum Tanner code \([[200,10,10]]\) introduced in~\cite{radebold2025explicit}; a Haah's cubic code \([[128, 14, 8]]\)~\cite{haah2011fractal}; and a fibre bundle (FB) code \([[126,8,9]]\), which we construct by following the methods in~\cite{hastings2021fiberbundlecodes}. We refer to these codes as HGP625, Tanner200, Haah128, and FB126, respectively. See \tab{qec_codes} for summary.

These codes possess complementary characteristics that make them well suited for benchmarking our approach. HGP625, Haah128 and FB126 exhibit a natural two-block structure in their parity-check matrices $H_X$ and $H_Z$, providing a natural left--right partition. To our knowledge, Tanner200 lacks such partitioning. It is also known that no single-ancilla CNOT-based SEC can reduce the circuit distance $d_{\mathrm{circ}}$ of HGP codes~\cite{Manes2025distancepreserving}. As a result, HGP625 serves as an ideal testbed for studying circuit depth and min-weight failure multiplicity $N_{\mathrm{fail}}$ under fixed $d_{\mathrm{circ}}$. Finally, Tanner200 features checks with widely varying and large weights, ranging from $6$ to $12$. Hence, Tanner200 should have large benefits from ancilla idling time optimisation.

For each code, we construct a high-performance LRC following the framework described in \sect{framework}. The partition stage differs based on whether the code has a natural left--right partition or not. For HGP625, Haah128 and FB126, we directly exploit their two-block structure, which yields near optimal depth of the corresponding LRCs. For Tanner200, we instead search for a suitable partition by considering many column permutations of the check matrices. This procedure identifies a partition with $t_1 = 9$ and $t_2 = 7$, resulting in an LRC depth of $t = t_1 + t_2 + 2 = 18$. Although this is far from the optimal depth $t_{\mathrm{min}} = 12$, it improves upon the naive circuit depth of $r(H_X) + r(H_Z) = 21$ derived from~\cite{TremblayEtAl2022} when allowing for overlapping SEC rounds.

Given the broad range of check weights in Tanner200, we additionally employ a modified edge-colouring algorithm based on linear sum assignment problem~\cite{Crouse2016assignment} that concentrates colour indices. Specifically, the algorithm produces an optimal edge-colouring in which colour $0$ is assigned to as many edges as possible, followed by colour $1$, and so on. This polynomial-time algorithm yields a more effective optimisation of the ancilla idling time $\tau_A(\mathcal{C})$ in the ranking stage of the framework. All remaining steps of the workflow are identical across the four codes. 
For the given example codes, we observed very small improvements when adding the optional step of evaluating the full extended-code distances of the top $r=5$ CNOT sequences.
While we did not employ this additional filtering step in the numerical analyses presented here, we believe it may prove useful for other codes or more complex circuit constructions.

After selecting the highest ranked CNOT sequence as deemed by the heuristic, we construct a noisy \texttt{Stim}~\cite{Gidney_2021} circuit according to \fig{lr-circuit}, with $d$ rounds of syndrome extraction. The resulting circuit is converted into a detector error model~\cite{Derks2025designingfault}, from which errors are sampled and decoded using BPOSD~\cite{roffe2020}. Full simulation details are provided in \app{sim-details}, and the results are plotted in \fig{data}.

\subsection{Comparison to other SEC constructors}

For comparison purposes, we construct alternative SECs using the \texttt{LDPC}~\cite{Roffe_LDPC_Python_tools_2022} package for all four codes, as well as the \texttt{QUITS}~\cite{kang2025quantum} package for HGP625 and FB126. To apply \texttt{QUITS} to FB126, we map the fibre bundle code to a quasi-cyclic lifted-product code and perform a randomised search over $10^4$ seeds to identify a circuit of depth $t = 8 + 2$; details of this mapping are provided in \app{FB-mapping}. For HGP625, we find a circuit of depth $t = 12 + 2$, matching the depth reported in~\cite{kang2025quantum}. For all constructions, we match the noise models and apply ancilla idling-time reduction (see \fig{ancilla_idling}) to both \texttt{LDPC} and \texttt{QUITS} circuits to ensure a fair comparison.
The results of alternative circuit constructors are plotted alongside the main results in \fig{data}. 

Let us further comment on the performance discrepancies observed between the different approaches. In all examples considered, the LRCs achieve smaller circuit depth than the alternative circuits. In most cases, this translates into substantially lower overall qubit-idling error rates. One exception is the Tanner200 code, whose LRC depth is comparable to that of the circuit generated by the \texttt{LDPC} package. Moreover, when the SEC follows an alternating-XZ structure (as in the \texttt{LDPC} circuits), ancilla initialisation and measurement for one check type can overlap with the CNOT operations of the other. This significantly reduces qubit idling time for the \texttt{LDPC} circuit, and therefore for Tanner200, the net improvement in idling time achieved by the LRC is marginal.

With respect to circuit distance, non-exhaustive distance tests indicate that the circuits generated by the \texttt{LDPC} package have lower distance than the corresponding LRCs for the Haah128 and FB126 codes. In contrast, the \texttt{LDPC} circuit for Tanner200 and the \texttt{QUITS} circuit for FB126 achieve the same circuit distance as their respective LRCs. As expected, all HGP625 circuits exhibit identical circuit distance equal to the code distance $d=8$. However, we observe a significantly worse residual distance profile for the $X$-type checks of the \texttt{QUITS} circuit compared to the other constructions. We believe this explains the high $Z$-logical observable error rate, which dominates the overall logical performance.

\vspace{5pt}
\noindent\textit{Note added:} During the preparation of this manuscript, two related works appeared on arXiv proposing alternative frameworks for SEC design, namely \textit{AlphaSyndrome}~\cite{liu2026alphasyndrometacklingsyndromemeasurement} and \textit{PropHunt}~\cite{viszlai2026prophuntautomatedoptimizationquantum}. In general, their techniques attain higher accuracy but at a much higher computational cost, thereby providing a complementary approach to ours. Both methods prioritise distance preservation, which we noted is often at the expense of increased circuit depth.

For comparison, we selected one representative circuit (typically with a larger qubit count) from each framework and compared it against the circuit generated by our method for the same code. We then evaluated and compared their logical performance under our noise model, as detailed in \app{alpha_prop}.
While the example circuit produced by \textit{PropHunt} achieved a larger circuit distance, this came at the cost of more than a threefold increase in circuit depth. Similarly, the circuit generated by \textit{AlphaSyndrome} exhibited a depth more than twice that of our corresponding construction. These large circuit depths lead to substantial drop in their logical performance.

\section{Discussion}
\label{sec:conclusion}

We have presented a simple, fast, and effective framework for constructing high-performance syndrome-extraction circuits (SECs) for arbitrary CSS codes. 
The resulting circuits are single-ancilla, CNOT-based SECs that, for many practically relevant quantum codes, achieve low depth and are optimised to preserve a large circuit distance $d_{\mathrm{circ}}$, minimise qubit idling time, and reduce the minimum-weight failure multiplicity $N_{\mathrm{fail}}$. 
Applying our framework to several distinct classes of LDPC codes, we observe consistent improvements in logical performance relative to existing SEC construction approaches.

As part of the framework, we establish several bounds on the circuit distance, corresponding to different settings and computational trade-offs. 
We expect that these theoretical results will have relevance beyond the scope of this work. 
A particularly interesting result is that interleaving cannot decrease the circuit distance.
More precisely, starting from a non-interleaved SEC, adjusting the timing of CNOTs while preserving their relative ordering within each check yields an interleaved circuit whose distance cannot decrease relative to the original non-interleaved SEC (but may increase).
This suggests that further performance gains may be achieved through clever interleaving strategies, which could serve as a foundation for future work on SEC design.

Another important direction is to extend the framework beyond static memory settings to codes with time-dependent structure, such as Floquet~\cite{hastings2021dynamically} and morphing codes~\cite{vasmer2022morphing,shaw2025lowering}, as well as to scenarios involving logical operations and code deformations~\cite{Horsman2012LatticeSurgery}. 
Adapting the left--right circuit structure to such dynamical codes will require rethinking how CNOT schedules are defined and coordinated across successive SEC rounds. 
Moreover, the overlap of consecutive rounds induced by the staggering of $X$ and $Z$-type checks may introduce additional challenges or overheads in practical scenarios. 
Exploring how to preserve circuit distance and logical performance under such dynamical conditions remains an important topic for future work.

\section{Acknowledgements}

AS thanks Oscar Higgott for fruitful discussions and for directing the authors to previous proposals on LRCs, and Balint Koczor for assistance with numerical simulations. AS and DEB were supported by the Engineering and Physical Sciences Research Council [grant number EP/Y004310/1]. AS thanks UKRI for the Future Leaders Fellowship Theory to Enable Practical Quantum Advantage (MR/Y015843/1). 
For the purpose of Open Access, the author has applied a CC BY public copyright licence to any Author Accepted Manuscript version arising from this submission.

\bibliography{biblio.bib}

\appendix

\section{LRC depth for lifted-product codes} 
\label{app:code-partitions}

In the following, we show that LRCs achieve near optimal depth for the class of lifted product (LP) codes.
LP codes are defined by parity check matrices
\begin{align}
\label{eq:LP-def}
H_X & = [C \otimes I_{m_D} \,|\, I_{m_C} \otimes D],\\
H_Z & = [I_{n_C} \otimes D^T \,|\, C^T \otimes I_{n_D}],
\label{eq:LP-def2}
\end{align}
where $C \in R^{m_C \times n_C}$ and $D \in R^{m_D \times n_D}$ are called \emph{protographs} with elements from some ring $R$, and $\otimes$ is the Kronecker product (over $R$). For constructing LRC, we use the block structure to partition LP into left/right data qubit sets, i.e., $L_X = C \otimes I_{m_D}$, $R_X = I_{m_C} \otimes D$, and so on. 

Consider the binary representations for any $a \in R$, denoted $\rho(a)\in \mathbb{F}_2^{l \times l}$.
LP codes can be quasi-cyclic, where $R = \mathbb{F}_2[x]/(x^l-1)$ and $\rho(a)$ is a circulant matrix, or, more generally, based on some group algebra $R = \mathbb{F}_2[G]$, in which case $\rho(a)$ is a combination of permutation matrices. For both, the crucial property we exploit is that the row- and column-weights of $\rho(a)$ are equal to some constant $w_a$.
This allows us to assign a weight to each element in each protograph $M \in \{C, D\}$ and construct a weight enumeration matrix $W\in \mathbb{Z}^{m_M \times n_M}$ with elements
\begin{equation}\nonumber
    W(M)_{ij} = w_{a_{ij}}
\end{equation}
where $a_{ij}$ is the ring element of $M_{ij}$.
With this, we denote the maximum row(column) weight of the protograph $M$ as the maximum of the sum of elements in any row(column) of $W$, i.e. $r(M) = \max_i\sum_j w_{a_{ij}}$ and $c(M) = \max_j\sum_i w_{a_{ij}}$, and $\delta(M) = \max (r(M), c(M))$. 
The rest of the proof for depth optimality (up to constants) for LRCs of LP codes follow in the analogous way to the HGP proof in the main text. Where instead of binary matrices $A$ and $B$ we use protographs $C$ and $D$ and their associated row(column) weights, and where the weight enumeration matrix $W(M)$ extends through the Kronecker products.

Note that if, for example, the row-weights of $\rho(a)$ were not all equal to some constant $w_a$, then we would not be able to define weight enumeration matrix $W(M)$ as above, and the rest of the proof would not generally hold.

Bivariate bicycle (BB) codes are special cases of LP codes where protographs $C$ and $D$ are of size $1\times 1$ and the ring $R$ is the group algebra of a product of two cyclic groups. Therefore, the same LRC depth arguments hold for BB codes.

\section{Construction of fibre bundle codes}
\label{app:fiber-bundle}

In the following, we introduce fibre bundle codes, describe their construction, and present the explicit code instance used in \sect{numerics}.

Fibre bundle (FB) codes are closely related to hypergraph-product codes (see \sect{LRC}) but incorporate an additional \emph{twisting} mechanism that can lead to an increased code distance. They are defined by parity check matrices 
\begin{align*}
H_X & = [B \otimes_\varphi I_{m_F} \,|\, I_{m_B} \otimes F],\\
H_Z & = [I_{n_B} \otimes F^T \,|\, B^T \otimes_\varphi I_{n_F}],
\end{align*}
for some choice of binary matrices $B \in \mathbb{F}_2^{m_B \times n_B}$ and $F \in \mathbb{F}_2^{m_F \times n_F}$, where the additional twisting mechanism is reflected by $\varphi$. Adopting terminology from topology, $B$ is the base code, $F$ is the fibre code and the fibre is twisted along the base according to some connection $\varphi$. 

Specifically, given any base matrix $B$, we have block matrices
\begin{align*}
(B \otimes_\varphi I_m)[i,j] = b_{ij}\varphi(i, j), \quad \varphi(i,j)\in \mathbb{F}_2^{l \times l}
\end{align*}
where $b_{ij}$ are binary elements of $B$ and $\varphi(i, j)$ is a matrix of size $m \times m$ representing some permutation action.
To ensure that FB have commuting checks, it is sufficient to have that each permutation commutes with the fibre code $F$, i.e. $\varphi_{ij}F = F \varphi_{ij}, \ \forall i, j$.

We construct a $[[126, 8, 9]]$ fiber bundle code (denoted FB126) as follows. We take the base code $B$ to be a $7 \times 7$ redundant parity-check matrix of the classical Hamming code $[7, 4, 3]$ such that the parity-check matrix $B^T$ also defines a code with parameters $[7, 4, 3]$. Akin to the original work, we choose the fibre to be a $d=9$ repetition code with periodic boundaries. This allows us to choose any power $k$ of a cyclic permutation as our twists (one can visualise them as $k \times 2\pi/9$ rotations of the circle graph representing the repetition code).
Specifically, we choose the following circular \emph{twist matrix}:
\begin{align}
T = 
\begin{bmatrix}
 1&  1& -1&  5&  0& -1& -1\\
 -1& -1&  4&  3&  2& -1&  3\\
 5& -1&  2& -1&  1&  0& -1\\
 7& -1& -1&  7& -1&  0&  4\\
 8&  5&  8& -1& -1& -1&  2\\
 -1&  6&  4&  3& -1&  1& -1\\
 -1&  8& -1& -1&  0&  6&  2
\end{bmatrix},
\label{eq:twist-matrix}
\end{align}
where each non-negative integer element corresponds to the power $k$ of the applied cyclic permutation. This yields a connection $\varphi(i, j) = S^{T_{ij}}$, where $S$ is a size $9\times9$ circular shift matrix, i.e., $S_{ij} = 1$ iff $i+1 = j$. The elements with value $-1$ indicate a zero element in the base code, and therefore, the corresponding block is the $0$ matrix.

\section{Mapping FB codes to QC-LP codes}
\label{app:FB-mapping}

To apply the \texttt{QUITS} framework to our fibre bundle (FB) code, we first map it to an equivalent lifted-product (LP) code. Owing to the circular twists used in our FB construction, this mapping yields a \emph{quasi-cyclic} LP code (see \app{code-partitions}).

Consider the LP code as defined by \eq{LP-def} and \eq{LP-def2} over the quotient ring $R= \mathbb{F}_2[x]/(x^l-1)$. Then we map
\begin{equation}
    b_{ij}\varphi(i, j)\to C_{ij},
\end{equation}
where each ring element $C_{ij} = b_{ij}x^{T_{ij}}$ following the twist matrix $T$ (\eq{twist-matrix}). If $b_{ij} = 0$, then $\rho(C_{ij}) = 0$, i.e. $C_{ij}$ is the zero element of the ring.

The finish the mapping we take $D$ to be a size $1\times 1$ protograph with $D_{11} = x^0 + x^1$, such that $F = \rho(D)$.

\section{Numerical distance estimation and upper bounds}
\label{app:bp-osd-upper-bounds}

In several places in this work we require upper bounds on the quantity
\begin{equation}
l_\text{min} = \min_{x \in \mathbb{F}_2^n} \left\{ |x| : H(x+y) = 0,\; Ax \neq 0 \right\},
\end{equation}
where $H \in \mathbb{F}_2^{M\times n}$ and $A \in \mathbb{F}_2^{k\times n}$ and $y \in \mathbb{F}_2^n$. 
In this work we consider two scenarios. 
In one scenario, $l_\text{min}$ is the residual distance of a residual error $y$ and $H$ and $A$ correspond to $H_X$ and $L_X$ (or $H_Z$ and $L_Z$).
In the other scenario, $y=0$ and $l_\text{min}$ corresponds to the extended-code distance for a residual error set $\mathcal{R}$, where $H$ and $A$ correspond to $H_X^{\mathrm{ext}}(\mathcal{R})$, $L_X^{\mathrm{ext}}(\mathcal{R})$ (or $H_Z^{\mathrm{ext}}(\mathcal{R})$, $L_Z^{\mathrm{ext}}(\mathcal{R})$).

We use two related randomised procedures to obtain upper bounds on $l_\text{min}$ by finding explicit bitstrings $x$ such that $H(x+y) = 0, Ax \neq 0$, and using their weight as an upper bound on $l_\text{min}$.

\subsection{Simple distance estimator}
\label{app:distance-upper-bound-basic}

The simplest approach, following~\cite{bravyiEtAl2024}, repeatedly constructs a random extension of $H$ and solves a single decoding instance using BPOSD.
We use this approach to estimate all residual distances for all codes in \sect{numerics} as follows.

\begin{enumerate}
\item Initialize $l_\text{ub} \gets \infty$.

\item For each trial $t=1,\dots,T$:
\begin{enumerate}
\item Sample a random vector $h$ from the row space of $H$ and independently sample a random nonzero vector $l$ from the row space of $A$.

\item Form $g = h + l$ and construct the extended matrix $H' = \begin{pmatrix} H \\ g \end{pmatrix}$.

\item Define the syndrome $\sigma \in \mathbb{F}_2^{M+1}$ with $\sigma_{M+1}=1$ and all other entries zero.

\item Use the BPOSD decoder for $H'$ to find a vector $x \in \mathbb{F}_2^n$ satisfying $H'(x+y)=\sigma$.

\item If $|x|<l_\text{ub}$, update $l_\text{ub} \gets |x|$ and set $t\gets 0$.
\end{enumerate}
\item Return $l_\text{ub}$ as an upper bound on $l_\text{min}$.
\end{enumerate}

The BPOSD settings we use for a rapid estimation of residual distances are the same as in \tab{BPOSD_full}, but with \emph{Max BP iterations} set to $200$. 
We use $T=100$ for each of the codes considered.
This procedure is simple and often sufficient for upper bounding residual distances due to the limited size of $H$.
Note that if one uses a provable minimum-weight decoder~\cite{ott2025decisiontreedecodersgeneralquantum,beni2025tesseractsearchbaseddecoderquantum,landahl2011fault} in place of a heuristic decoder such as BP-OSD, this algorithm can be modified to find the distance exactly, but may beome prohibitively slow.

\subsection{Adaptive-prior distance estimator}
\label{app:distance-upper-bound-advanced}

For difficult instances we use a more aggressive approach introduced in \cite{beverland2025failfasttechniquesprobe} that repeatedly perturbs the decoder's prior probabilities in order to bias the search toward promising regions of the solution space. 
This is used for all extended-code distance estimations for the gross code.
The overall structure remains the same as for the simple estimator: each iteration constructs a random extension of $H$ and attempts to decode the syndrome corresponding to the appended row.

\begin{enumerate}
\item Fix the number of random-row trials $T_{\mathrm{row}}$ and prior perturbations $T_{\mathrm{prior}}$, and initialize $l_\text{ub}\gets\infty$.

\item For each row trial $t=1,\dots,T_{\mathrm{row}}$:
\begin{enumerate}
\item Sample $h$ from the row space of $H$ and a random nonzero $l$ from the row space of $A$, set $g=h+l$, and form $H'=\begin{pmatrix}H\\g\end{pmatrix}$.

\item Define the syndrome $\sigma\in\mathbb{F}_2^{M+1}$ with $\sigma_{M+1}=1$ and all other entries zero, and compute the support of $g$.

\item For each prior trial $i=1,\dots,T_{\mathrm{prior}}$:
\begin{enumerate}
\item Generate channel probabilities for the decoder near a small base error rate and increase the probabilities on a random subset of $\mathrm{supp}(g)$ to bias the search.

\item Run the BPOSD decoder on $(H',\sigma)$ to obtain a candidate $x$.

\item If $H'(x+y)=\sigma$, compute $|x|$ and update $l_\text{ub}\gets |x|$ if this weight is smaller.

\item Update $l_\text{ub}$ whenever a smaller-weight solution is found.
\end{enumerate}
\end{enumerate}

\item Return $l_\text{ub}$ as an upper bound on $l_\text{min}$.
\end{enumerate}

\section{Circuit simulation details}
\label{app:sim-details}

We estimate the logical error rate of the QEC system by estimating the logical error rate for $Z$ (and $X$)-type observables separately. We construct \texttt{Stim}~\cite{Gidney_2021} circuits where all data qubits are initialised in $\ket{0}$ (respectively $\ket{+}$) states such that the logical observables corresponding to logical qubits are initialised in $\overline{\ket{0}}$ (respectively $\overline{\ket{+}}$). At the end of the circuit we measure all data qubits along the same basis as initialised, and this measurement data provide us information to reconstruct a perfect round of syndrome extraction (for the specific basis). Logical error rate for each observable type $P_Z(p)$ and $P_X(p)$ is estimated based on sampling as described in the main text, and final, reported logical error rate is approximated as $P(p) = P_Z(p) + P_X(p)$.

Using \texttt{Stim} and \texttt{LDPC} packages, each circuit was converted to the detector error model, sampled and decoded using belief-propagation ordered-statistics decoder BPOSD~\cite{roffe2020}.
\tab{BPOSD_full} lists the BPOSD parameters used in the simulation, other parameters were set to  default values.

\begin{table}[h]
\centering
\caption{BPOSD parameters}
\label{table:bposd-params}
\begin{tabular}{ll}
\toprule
Parameter & Value \\
\midrule
BP method & Min-sum (MS) \\
Max BP iterations & $10^4$ \\
MS scaling factor & $0$ \\
OSD method & OSD-CS \\
OSD order & $7$ \\
\bottomrule
\end{tabular}
\label{table:BPOSD_full}
\end{table}

\section{Comparison with AlphaSyndrome and PropHunt}
\label{app:alpha_prop}

To briefly investigate the recently proposed frameworks \textit{AlphaSyndrome}~\cite{liu2026alphasyndrometacklingsyndromemeasurement} and \textit{PropHunt}~\cite{viszlai2026prophuntautomatedoptimizationquantum}, we consider their respective circuits for the $[[61,1,9]]$ hexagonal colour code (from \textit{AlphaSyndrome}) and the $[[108,18,4]]$ random quantum Tanner code (from \textit{PropHunt}), and compare them to the corresponding circuits generated by our framework. Their logical performance was simulated under the same noise model and simulation settings used for the other codes in this work (see \sect{SEC} and \app{sim-details}). The results are shown in Fig.~\fig{alpha_prop}. As expected, the substantial increase in circuit depth offsets the potential gain in circuit distance under realistic error-rate regimes.

\begin{figure}[tbp]
    \centering
    \includegraphics[width =1\linewidth]{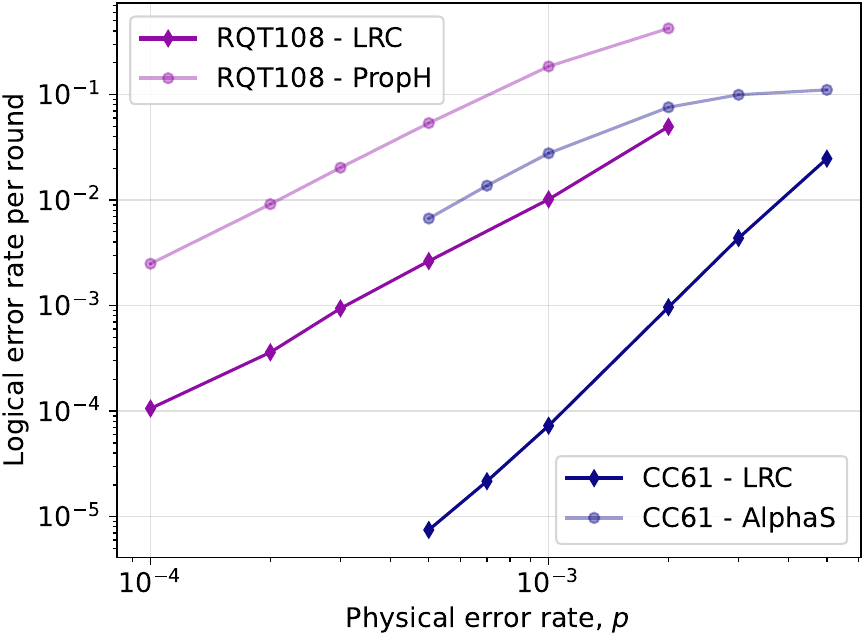}
    \caption{
    Comparison of the logical performance of circuits for the $[[61,1,9]]$ colour code (CC61) and the $[[108,18,4]]$ random quantum Tanner code (RQT108). The circuits were generated using our approach (LRC) and the recently proposed frameworks \textit{AlphaSyndrome}~\cite{liu2026alphasyndrometacklingsyndromemeasurement} (AlphaS) and \textit{PropHunt}~\cite{viszlai2026prophuntautomatedoptimizationquantum} (PropH), respectively.
    }
    \label{fig:alpha_prop}
\end{figure}

\section{Residual profile analysis}
\label{app:Residual_analysis}

\begin{figure}[tbp]
    \centering
    \includegraphics[width =1\linewidth]{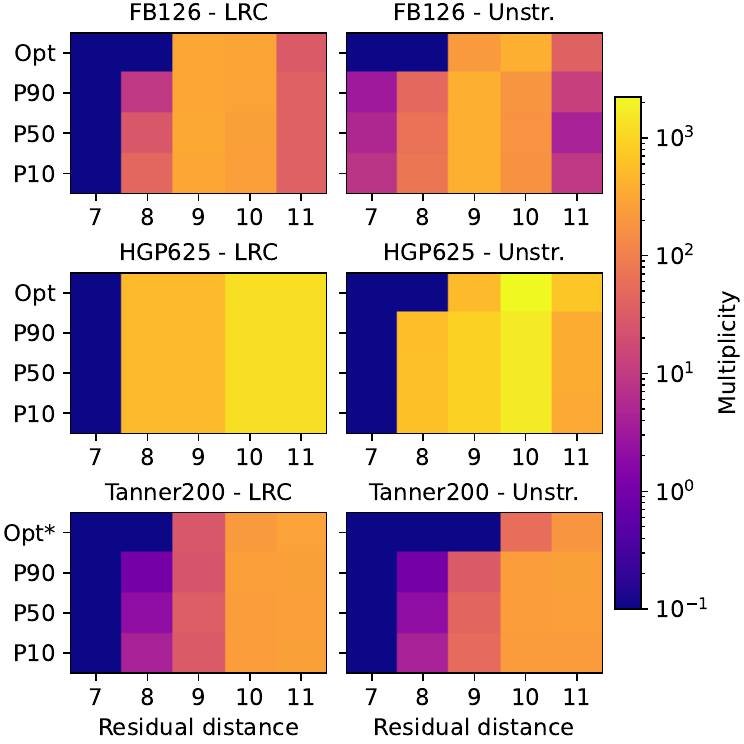}
    \caption{
    Representative residual distance profiles corresponding to the 10th, 50th, and 90th percentiles of the lexicographic ranking, together with the optimal residual distance profile, shown for each code and for both CNOT sequence sets -- the set of all left--right sequences (LRC) and the set of all unstructured CNOT sequences (Unstr.). No sequence in either set exhibits a residual distance $\Delta_\text{min}(\mathcal{C}) < 7$.
    }
    \label{fig:res-analysis}
\end{figure}

The benefits of short circuit depth $t$ and low ancilla idling time $\tau_A(\mathcal{C})$ are clear -- fewer idling qubits lead to reduced circuit noise and thus better performance. To assess the importance of the residual error profiling, we begin by examining the range of residual distance profiles obtained by LRCs. Afterwards, we investigate how well the circuit residual distance $\Delta_\text{min}(\mathcal{C})$ and their profiles $W(\mathcal{C})$ translate into QEC system's metrics such as the circuit distance $d_{\mathrm{circ}}$ and the min-weight failure multiplicity $N_{\mathrm{fail}}$.

\subsection{Residual distance profile distribution}

\begin{table}[t]
\centering
\scalebox{0.90}{%
\renewcommand{\arraystretch}{1.2}
\setlength{\arrayrulewidth}{0.4pt}

\begin{tabular}{|c|c?c|c?c|c|}
\cline{3-6}
\multicolumn{2}{c?}{} & \multicolumn{2}{c?}{\textbf{LRC}} & \multicolumn{2}{c|}{\textbf{Unstructured}} \\
\cline{3-6}
\multicolumn{2}{c?}{} & \textbf{Opt} & \textbf{P50} & \textbf{Opt} & \textbf{P50} \\
\thickhline

\multirow{2}{*}{\textbf{FB126}} & $d_{\mathrm{est}}$ &
\textbf{8} & 8 & 7 & 6 \\
\cline{2-6}
& $N_{\mathrm{est}}$ &
$\mathbf{3.10 \times 10^3}$ & $2.56 \times 10^4$ & $1.48 \times 10^4$ & $1.81 \times 10^2 $ \\
\thickhline

\multirow{2}{*}{\textbf{HGP625}} & $d_{\mathrm{est}}$ &
8 & 8 & \textbf{8} & 8 \\
\cline{2-6}
& $N_{\mathrm{est}}$ &
$6.13 \times 10^5$ & $6.14 \times 10^5$ & $\mathbf{4.27 \times 10^5}$ & $6 \times 10^5$ \\
\thickhline

\multirow{2}{*}{\textbf{Tanner200}} & $d_{\mathrm{est}}$ &
7 & 7 & \textbf{9} & 7 \\
\cline{2-6}
& $N_{\mathrm{est}}$ &
$1.22 \times 10^4$ & $1.51 \times 10^4$ & $\mathbf{9.08 \times 10^5}$ & $2.11 \times 10^4$ \\
\thickhline
\end{tabular}%
}
\caption{Estimated circuit distance $d_{\mathrm{est}}$ and min-weight failure multiplicity $N_{\mathrm{fail}}$ for optimal and 50th percentile circuit representatives of both $S_{LRC}$ (LRC) and $S_U$ (Unstructured) sets of circuits. In the limit of $p\to 0$, these parameters estimate the logical error rate as $P_{\mathrm{est}}(p) = N_{\mathrm{est}}\,p^{\lceil{d_{\text{est}}/2\rceil}}$. Best estimate for each code is highlighted in bold.}
\label{table:residual_analysis}
\end{table}

We begin by comparing the residual distance profiles produced by min-depth LRCs to those attainable through arbitrary CNOT sequencing. For codes Tanner200, HGP625 and FB126, we sample up to $1000$ CNOT sequences from the set of all sequences possible for LRCs, denoted $S_{\mathrm{LRC}}$, who differ by permutations of the same depth-optimal edge-colouring. We then rank them according to their residual profiles. The sequences whose profile is lexicographically smaller are ranked higher.  In a similar fashion, we sample $1000$ CNOT sequences from a set of \emph{unstructured} sequences, $S_U$, in which the CNOT sequence for each check is arbitrary. These sequences are also ranked using the same criterion.

For both $S_{\mathrm{LRC}}$ and $S_U$, we additionally identify a representative CNOT sequence with the optimal (Opt) residual distance profile. Since $S_{\mathrm{LRC}} \subset S_U$, the optimal unstructured sequence is guaranteed to achieve a profile that is at least as good as that of the optimal LRC. For Tanner200, however, the high check weights make an exhaustive search infeasible. In this case, we approximate the optimal circuits as follows: for $S_{\mathrm{LRC}}$, we select the highest-ranked CNOT sequence among the $1000$ sampled; while for $S_U$, we sample up to $9!$ permutations of each check’s CNOT sequence and assemble the total CNOT sequence using the best-performing permutation for each check, as determined by the check's residual distance profile.

\fig{res-analysis} shows representative residual distance profiles corresponding to the 10th, 50th, and 90th percentiles of the ranking, as well as the optimal profile, for each code and CNOT sequence class. Notably, the optimal LRC sequence achieves a profile that is at least as good as that of the 90th-percentile unstructured sequence. This suggests that, with respect to CNOT sequencing quality as measured by the circuit's residual distance profile, the set $S_{\mathrm{LRC}}$ consistently contains above-average circuits.

\subsection{Logical performance}

Next, we investigate how well the residual distance profiles of representative CNOT sequences correlate with the circuit distance $d_{\mathrm{circ}}$ and the min-weight failure multiplicity $N_{\mathrm{fail}}$. To this end, we construct single-round SEC circuits (followed by an additional perfect round) for the representative CNOT sequence of the 50th percentile (P50) and the optimal (Opt) profile for each code.

Because unstructured CNOT sequences generally do not admit low-depth implementations, we implement a sequential execution of checks for all circuits to enable a fair comparison. Specifically, all checks are implemented one at a time, with idling noise set to zero -- the CNOTs corresponding to check $z_1$ are executed first, followed by those for $z_2$, and so on, until the final check $x_m$. This construction allows us to investigate the effect of CNOT sequencing in isolation.
Each circuit is then converted into a detector error model, and the techniques introduced in~\cite{beverland2025failfasttechniquesprobe} are used to bound both the circuit distance and the min-weight failure multiplicity.

Using the gathered data, the circuit distance is upper bounded as
\begin{equation}
\nonumber
    d_{\mathrm{circ}} \le d_{\mathrm{est}} = \min\!\left(d_{\mathrm{est}}^{(X)},\, d_{\mathrm{est}}^{(Z)}\right),
\end{equation}
where $d_{\mathrm{est}}$ denotes the estimated distance, and the superscripts $X$ and $Z$ indicate the type of logical observables being tested. 
Similarly, we obtain a lower bound on $N_{\mathrm{fail}}$ associated with the leading-order term of the logical failure probability. Assuming the order of this term is $p^{\lceil d_{\mathrm{est}}/2\rceil}$, we have
\begin{equation}
\nonumber
    N_{\mathrm{fail}} \ge N_{\mathrm{est}} = N_{\mathrm{est}}^{(X)} + N_{\mathrm{est}}^{(Z)}.
\end{equation}
In the limit \(p \to 0\), these quantities define an estimate of the logical error rate as
\begin{equation}
    P_{\mathrm{est}}(p) = N_{\mathrm{est}}\,p^{\lceil d_{\mathrm{est}}/2\rceil}. \nonumber
\end{equation}

The resulting values of $d_{\mathrm{est}}$ and $N_{\mathrm{est}}$ for each representative circuit are reported in \tab{residual_analysis}.

We find that, with the exception of a single case, residual distance profiles serve as reliable indicators of the relative logical performance of circuits in the low-$p$ regime. The sole outlier is the optimal unstructured circuit for the FB126 code. Although its profile attains the best rank, it exhibits a lower estimated $d_{\mathrm{circ}}$ than several alternative circuits. This discrepancy is not surprising -- residual profiles offer a fast and inexpensive means of analysis individual CNOT sequences, but they do not account for interactions between residual errors arising from different checks. This limitation suggests a natural extension of the present work, incorporating more computationally intensive yet more comprehensive analyses of CNOT sequences.
We also note that improvements in $d_{\mathrm{circ}}$ and $N_{\mathrm{fail}}$ do not necessarily translate directly to better QEC system's performance under realistic decoding strategies and error rates. Nevertheless, such metrics can facilitate the discovery of high-performance circuits.

\subsection{Syndrome weight}

Finally, we observe that for codes with a natural left--right partitioning, such as HGP and FB codes, the compressed detector error model check matrices of LRCs exhibit a higher density of non-zero elements than those obtained from other SECs. This increased density is also reflected in a larger number of detection events per shot during circuit simulation, which in turn slows down decoding. 
As a concrete example, for the two-round SEC LRCs shown above (Opt and P50), we observe an average increase of approximately $8.4\%$ for FB126 and $14\%$ for HGP625 in the check-matrix density relative to the corresponding unstructured circuits. Consequently, the average BPOSD decoding time is increased by a factor of approximately $\times1.3$ for FB126 and $\times2$ for HGP625.

We speculate that this behaviour arises from large overlapping detection regions on ancilla qubits, which result from the manner in which check commutation is achieved in two-block codes. While such large overlap improves the detection of harmful hook errors, it also results in a higher average number of detection events. We leave a more detailed investigation of this trade-off to future work.

\section{Analysis of gross code circuits}
\label{app:gross-analysis}

The gross code (see in \fig{gross-code}) has code distance $12$, yet to date, no stabiliser-extraction circuit has been discovered with circuit distance greater than $10$.
Here we use the techniques developed in the main text to study the SECs for the gross code in detail.

\begin{figure*}[t]
  \centering
  \includegraphics[width=0.8\linewidth]{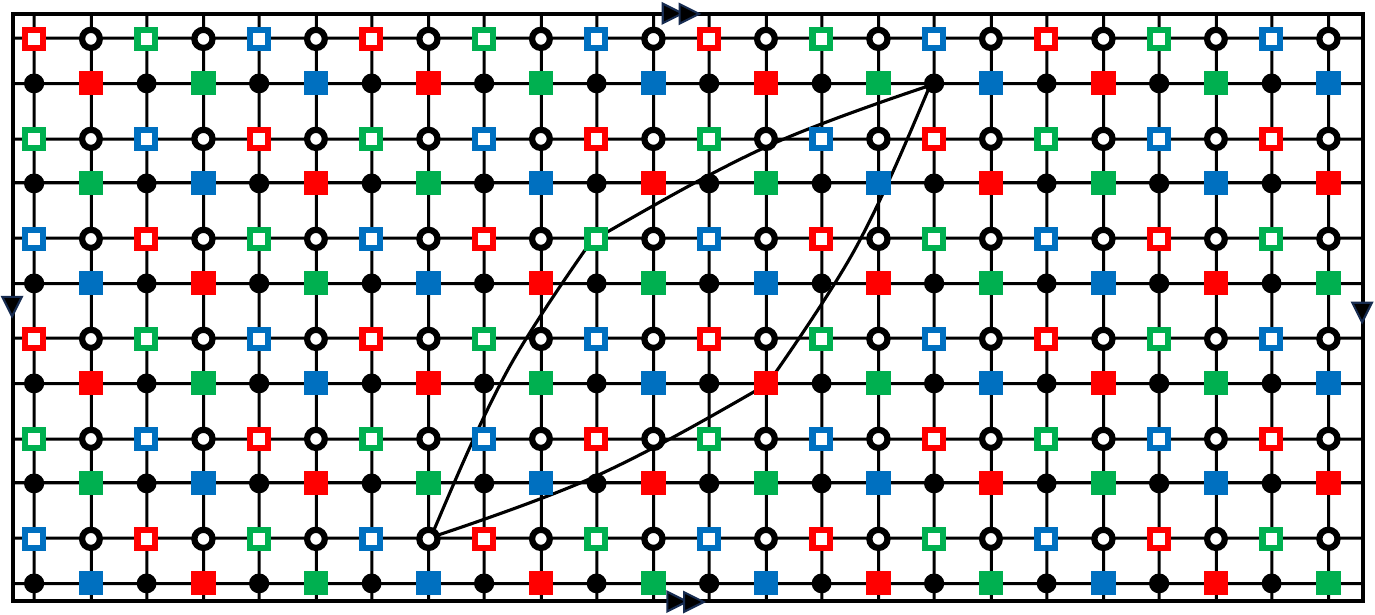}
  \caption{Gross code bipartite Tanner graph. 
  Circles are data qubits. 
  Filled squares are Z-type checks, empty squares are X-type checks. 
  There are 6 data qubits in the support of each check, with nearest neighbours marked by horizontal and vertical edges, and with two diagonal long-range edges which are only shown one check of each type (the others are obtained by translations of these). 
  There is a natural three-colouring of checks of each type such that two checks that share support have distinct colours.
  We will use this three-colouring to specify some families of stabiliser extraction circuits.}
  \label{fig:gross-code}
\end{figure*}

\subsection{Eliminating gross-code circuits}

Our goal is to identify a gross code SEC with the highest possible circuit distance.
\propos{all-circuits-bound} in the main text shows that, while one may initially have hoped to obtain a non-interleaved SEC with circuit distance 12 for the gross code (matching the code distance), no such circuit exists.
Given this, we here seek a non-interleaved SEC with circuit distance 11.
To do this, we use our extended-code distance upper-bound methods from \sect{Freedom} to progressively narrow the space of viable candidate circuits by proving a sequence of propositions.

Many SECs (such as the standard one from Ref.~\cite{bravyiEtAl2024}) involve all X checks performing the same operations uniformly, but translated with respect to one another (given the symmetry of the code). 
Unfortunately, the following proposition tells us that such simple classes of uniformly tiled non-interleaved circuits are too restricted to give a circuit distance of 11.

\begin{proposition}[Uniformly-tiled gross circuits]
\label{propos:uniform-tiling-bound}
Let $d_{\mathrm{circ}}$ be the circuit distance of any uniformly-tiled non-interleaved SEC for the gross code, by which we mean that the local restriction of the circuit to each X check in \fig{gross-code} is the same. 
Then,
\begin{equation}
d_{\mathrm{circ}} \ \le\ 10.
\end{equation}
\end{proposition}

\begin{proof}
The proof is similar to that of \propos{all-circuits-bound}, but now instead of $\mathcal{R}_i$ corresponding to three errors that arise for a single $X$-type check, it corresponds to the union of three-error sets that arise when the same set is tiled across all 72 $X$-type checks (such that $\mathcal{R}_i$ contains $72 \times 3$ elements). 
As explained in the proof of \propos{all-circuits-bound}, there are 90 inequivalent residual error sets for a weight-six check, such that $i =1,2,\dots,90$.

Any uniformly-tiled non-interleaved SEC must have a residual-error set $\mathcal{E}$ that includes one of these sets $\mathcal{R}_i$ as a subset, and as such \thm{circ_dist_bounds} implies $d_{\mathrm{circ}} \le \max_i d_{\mathrm{ext}}(\mathcal{R}_i)$.
We find that of the extended-code distance is at most 10 for all of the 90 possible residual error sets $\mathcal{R}_i$, proving that $d_{\mathrm{circ}} \leq 10$ as required. 
\end{proof}

Given that uniformly tiled non-interleaved SECs cannot achieve circuit distance 11, we consider non-uniform tilings. 
Given that each of the $X$-type and the $Z$-type Tanner graphs are three-colourable (as described in \fig{gross-code}) we seek circuits which are uniformly tiled for checks of the same colour. 
We can then exclude almost all possibilities within this class as follows.

\begin{proposition}[Three-colour tiled gross circuits]
\label{propos:3colour-tiling-bound}
Consider non-interleaved circuits with the same relative ordering on all checks with same colour.
For the part of the circuit measuring Z checks (respectively X checks), there are at most 3 relative orderings of three-colour tiled circuits with circuit distance 11.
These 3 orderings are the same up to a cyclic permutation of colours.
\end{proposition}

\begin{proof}
We start by considering all 90 of the residual sets for an isolated Z-type check.
We exclude those with extended-code distance upper-bounded by 10, since they cannot participate in any circuit having extended-code distance 11, leaving 45 residual sets. 

Next, we tile all red checks with a residual set, and upper bound the resulting extended-code distances.
Scanning through the 45 residual sets, we find all but 4 have extended-code distance of 10 or below.

We apply the same approach separately to the blue checks and to the red checks, finding that precisely the same 4 residual sets that survived the screening when tiled on the red checks survive when tiled on all checks of another colour too. 

Next, we consider all combinations of these 4 error sets, applying the same error set to all checks of a given colour.
Of the $4^3=64$ combinations, we find that all but 3 have extended-code distance of 10 or below. 
\end{proof}

We ultimately obtain only 3 tiling combinations of residual errors for which the extended-code distance upper bound was 11, with all other three-coloured tilings having been eliminated as they result in circuit distance at most 10. 
We show one of the three in \fig{3-coloured-residuals}, the others are cyclic permutations of the colours.

\begin{figure}[t]
    \centering
    \includegraphics[width=0.8\linewidth]{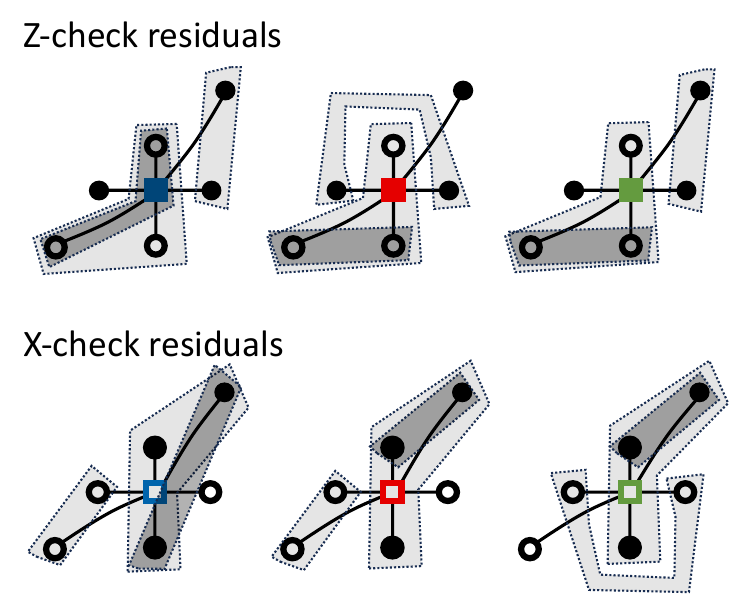}
    \caption{Residual error patterns for a three-colour tiling circuit with extended-code distance upper bound $11$.
    Each of the X- and Z-type residual patterns can be independently cyclically permuted across the
    colour classes (B$\rightarrow$R$\rightarrow$G) without reducing this bound. All other patterns
    ($90^3 - 3 = 728997$ for each of X and Z) have extended-code distance $\le 10$.
    }
    \label{fig:3-coloured-residuals}
\end{figure}

\subsection{An explicit 3-coloured circuit}

In \fig{3-coloured-circuit-compressed} we show an explicit schedule that achieves the residual error pattern in \fig{3-coloured-residuals}.
We have tested this circuit extensively by searching for minimum-weight logical operators, conjecture that this has circuit distance 11.
The justification of this conjecture is based on using the same technique to find and collect minimum-weight logical operators for a single QEC cycle of the standard circuit and for our circuit.
While a weight-10 logical was discovered after a single outer loop of the search, no min-weight logical operators of weight below 11 for our circuit after ten thousand outer loops (see \app{distance-upper-bound-advanced}).

\begin{figure}[t]
    \centering
    \includegraphics[width=0.8\linewidth]{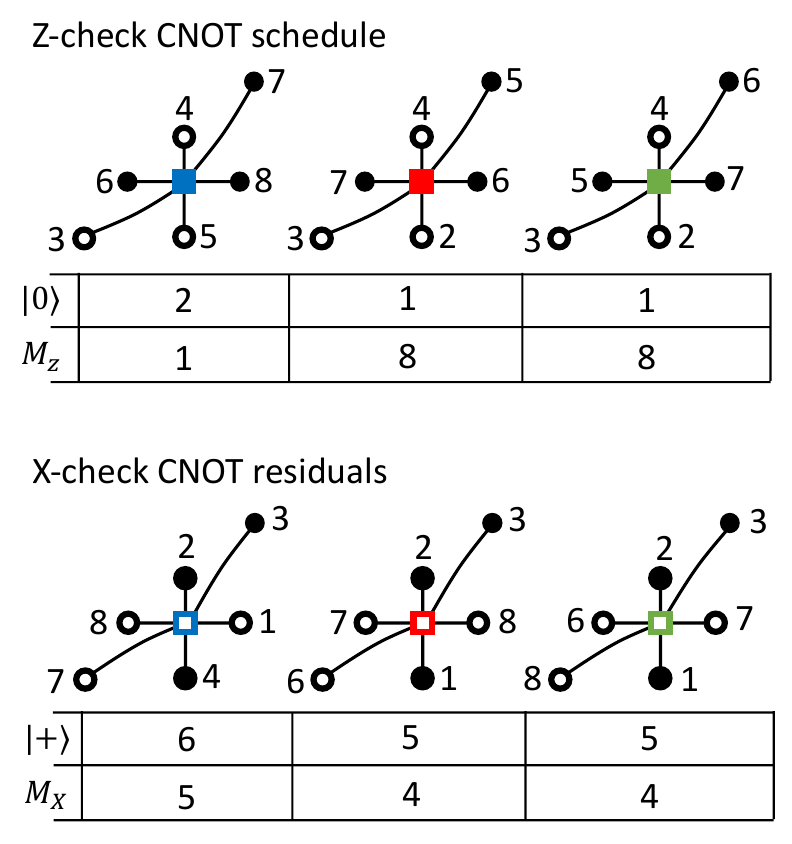}
    \caption{
    A 3-colour-tiled left--right schedule realizing the residual error pattern in \fig{3-coloured-residuals}.
    The schedule has depth 8, matching the original gross-code circuit.
    Its extended-code distance is upper bounded by 11, and we conjecture this bound is tight.
    For each check, the neighbouring data qubits are labelled by their CNOT time steps; the preparation and measurement times are shown below.
    }
    \label{fig:3-coloured-circuit-compressed}
\end{figure}

We could further investigate non-staggered, interleaved variants of the circuit, which one could hope might increase the circuit distance.
However, our motivation to do this is reduced by the very nice recent paper on morphing bivariate bicycle codes which produces circuits involving the same number of physical and logical qubits as the gross code, but with circuit distance twelve Ref.~\cite{shaw2025lowering}.

\end{document}